\documentclass[letterpaper]{article} 
\usepackage[submission]{aaai25}  
\usepackage{times}  
\usepackage{helvet}  
\usepackage{courier}  
\usepackage[hyphens]{url}  
\usepackage{graphicx} 
\usepackage{natbib}  
\usepackage{caption} 
\usepackage{algorithm}
\usepackage{algorithmic}

\usepackage{newfloat}
\usepackage{listings}
\DeclareCaptionStyle{ruled}{labelfont=normalfont,labelsep=colon,strut=off} 
\floatstyle{ruled}
\newfloat{listing}{tb}{lst}{}
\floatname{listing}{Listing}
\usepackage[dvipsnames]{xcolor}
\usepackage{xspace}
\usepackage{fancyhdr}
\usepackage{tcolorbox}
\usepackage{graphicx}

\usepackage{amssymb}
\usepackage{amsthm}
\usepackage{amsmath,mathtools}
\DeclarePairedDelimiter{\ceil}{\lceil}{\rceil}
\usepackage{stmaryrd} 
\usepackage{stackengine}
\usepackage[english,ruled,vlined,algo2e,linesnumbered]{algorithm2e}

\usepackage{multirow}
\usepackage{tabularray}
\usepackage{colortbl}

\usepackage{anyfontsize}

\usepackage{thm-restate}
\usepackage{cleveref}

\newtheorem{theorem}{Theorem}[section]
\newtheorem{definition}[theorem]{Definition}

\newtheorem{lemma}[theorem]{Lemma}
\newtheorem{corollary}[theorem]{Corollary}
\newtheorem{observation}[theorem]{Observation}
\newtheorem*{observation*}{Observation}

\newtheorem{proposition}[theorem]{Proposition}

\newcommand{\bigparagraph}[1]{\vspace{0.4em}\noindent\textbf{#1}}

\newcommand{\ie}{i.\,e.,\xspace}

\newcommand{\jcal}{\ensuremath{\mathcal{J}}\xspace}
\newcommand{\scal}{\ensuremath{\mathcal{S}}\xspace}

\newcommand{\pcal}{\ensuremath{\mathcal{P}}\xspace}
\newcommand{\gcal}{\ensuremath{\mathcal{G}}\xspace}
\newcommand{\ecal}{\ensuremath{\mathcal{E}}\xspace}

\newcommand{\wcal}{\ensuremath{\mathcal{W}}\xspace}

\newcommand{\tuple}[1]{\ensuremath{\langle {#1} \rangle}\xspace}

\newcommand{\toutdegree}{\ensuremath{\delta^{out}}}
\newcommand{\tindegree}{\ensuremath{\delta^{in}}}
\newcommand{\tdegree}{\ensuremath{\delta}}
\newcommand{\edgestream}{\ensuremath{E_S}\xspace}
\newcommand{\tdirect}[3]{\ensuremath{((#1,#2),#3)}}
\newcommand{\tundirect}[3]{\ensuremath{(\{#1,#2\},#3)}}

\newcommand{\tmax}{\ensuremath{t_{\max}}\xspace}

\newcommand{\eecShort}{eec\xspace}

\newcommand{\SWEC}{\textsc{SwalkEEC}\xspace} 

\newcommand{\TJEECwithK}[2]{$#1$-\textsc{#2T$\Phi$EEC}}
\newcommand{\TPEECwithK}[2]{$#1$-\textsc{#2TpathEEC}}
\newcommand{\TTEECwithK}[2]{$#1$-\textsc{#2TtrailEEC}}
\newcommand{\TWEECwithK}[2]{$#1$-\textsc{#2TwalkEEC}}
\newcommand{\TJEECnoK}[1]{\textsc{#1T$\Phi$EEC}}
\newcommand{\TPEECnoK}[1]{\textsc{#1TpathEEC}}
\newcommand{\TTEECnoK}[1]{\textsc{#1TtrailEEC}}
\newcommand{\TWEECnoK}[1]{\textsc{#1TwalkEEC}}

\newcommand{\setEndTs}[1]{\ensuremath{V_{\text{end}}(#1)}\xspace}

\newcommand{\TerminalStartSet}{\ensuremath{V_S}\xspace}
\newcommand{\TerminalEndSet}{\ensuremath{V_E}\xspace}

\DeclareMathOperator{\walks}{walks}

\newcommand{\NP}{\ensuremath{\mathtt{NP}}\xspace}

\newcommand{\Wtwo}{\ensuremath{\mathtt{W}[2]}\xspace}
\newcommand{\PP}{\ensuremath{\mathtt{P}}\xspace}

\newcommand{\FPT}{\ensuremath{\mathtt{FPT}}\xspace}
\newcommand{\XP}{\ensuremath{\mathtt{XP}}\xspace}

\newcommand{\bigoh}{\mathcal{O}}

\newcommand{\poly}{\ensuremath{\mathtt{poly}}\xspace}

\newcommand{\hittingset}{\textsc{HittingSet}\xspace}

\newcommand{\sat}[1]{$#1$-\textsc{SAT}}
\newcommand{\naesat}[1]{\textsc{NAE}$(#1)$\textsc{SAT}\xspace}
\newcommand{\xorsat}[1]{\textsc{XOR$(#1)$SAT}\xspace}

\usepackage{tabularx}
\newcommand{\problemtitle}[1]{\gdef\@problemtitle{#1}}
\newcommand{\probleminput}[1]{\gdef\@probleminput{#1}}
\newcommand{\problemquestion}[1]{\gdef\@problemquestion{#1}}

\usepackage[textsize=footnotesize]{todonotes}

\newif\iflong
\newif\ifshort

\longtrue

\iflong
\else
\shorttrue
\fi

\title{How Many Lines to Paint the City: Exact Edge-Cover in Temporal Graphs}
\author{
    Argyrios Deligkas\textsuperscript{\rm 1},
    Michelle D\"oring\textsuperscript{\rm 2},
    Eduard Eiben\textsuperscript{\rm 1},
    Tiger-Lily Goldsmith\textsuperscript{\rm 1},
    George Skretas\textsuperscript{\rm 2},
    Georg Tennigkeit\textsuperscript{\rm 2}
}
\affiliations {
    \textsuperscript{\rm 1}Royal Holloway, University of London, Egham, United Kingdom\\
    \textsuperscript{\rm 2}Hasso Plattner Institute, University of Potsdam, Potsdam, Germany\\
    argyrios.deligkas@rhul.ac.uk,
    michelle.doering@hpi.de,
    eduard.eiben@rhul.ac.uk,
    tiger-lily.goldsmith@rhul.ac.uk,
    georgios.skretas@hpi.de,
    georg.tennigkeit@hpi.de
}

\begin{document}

\maketitle

\begin{abstract}
Logistics and transportation networks require a large amount of resources to realise necessary connections between locations and minimizing these resources is a vital aspect of planning research. 
Since such networks have dynamic connections that are only available at specific times, intricate models are needed to portray them accurately. 
In this paper, we study the problem of minimizing the number of resources needed to realise a dynamic network, using the temporal graphs model.
In a temporal graph, edges appear at specific points in time.
Given a temporal graph and a natural number $k$, we ask whether we can cover every temporal edge exactly once using at most $k$ temporal journeys; in a temporal journey consecutive edges have to adhere to the order of time.
We conduct a thorough investigation of the complexity of the problem with respect to four dimensions: 
(a) whether the type of the temporal journey is a walk, a trail, or a path; 
(b) whether the chronological order of edges in the journey is strict or non-strict; 
(c) whether the temporal graph is directed or undirected; 
(d) whether the start and end points of each journey are given or not.
We almost completely resolve the complexity of all these problems and provide dichotomies for each one of them with respect to $k$.
\end{abstract}

\newcommand{\szTCref}[1]{{\scriptsize#1}}
\newcommand{\szTk}[1]{{\scriptsize#1}}
\definecolor{Shilo}{rgb}{0.905,0.705,0.709}
\definecolor{VanillaIce}{rgb}{0.956,0.835,0.843}
\definecolor{MexicanRed}{rgb}{0.678,0.149,0.172}
\definecolor{White}{rgb}{1,1,1}

\newif\ifktwo
\ktwofalse

\begin{table*}[t]
\centering
\footnotesize
\begin{tblr}{
    width=\textwidth,
    colspec = {@{}Q[50]Q[40]Q[55]Q[55]Q[60]Q[80]},
    row{even} = {Shilo},
    row{odd} = {VanillaIce},
    row{1,7,8} = {White},
    column{1,2} = {White},
    cell{2}{1} = {r=2}{},
    cell{4}{1} = {r=2}{},
    vline{4-6} = {2-5}{white},
    hline{1-2,4,6} = {-}{MexicanRed},
    hline{3,5} = {2-6}{white}    
}
\SetCell[c=2]{l} Complexity && \centering Paths  & \centering Trails & \centering Walks & \centering  Walks with fixed terminals\\
    directed edges    
    & strict   
        & \NP-c \szTCref{(\Cref{thm:SD_TPD})}
        & \NP-c \szTCref{(\Cref{thm:TTD_para})}   
        & \PP\ \szTCref{(\Cref{thm:SD_TWD_linear_time})}   
        & \PP\ \szTCref{(\Cref{thm:SD_TWD_linear_time})}   
    \\
    & non-strict 
        & \NP-c \szTCref{(\Cref{thm:SD_TPD})}
        & \NP-c \szTCref{(\Cref{thm:TTD_para})}
        & W[2]-h \szTCref{(\Cref{thm:DNS_TWD})}
        & \PP\ \szTCref{(\Cref{thm:DNS_TWD_fixed_terminals})}
    \\
    undirected edges     
    & strict
        & \NP-c \szTCref{(\Cref{thm:SD_TPD})}
        & \NP-c \szTCref{(\Cref{thm:TTD_para})}
        & \NP-c \szTCref{(\Cref{thm:NSUD_TWD})}
        & $\PP\footnotemark[1]$ \szTCref{(\Cref{cor:UD_S_fxdTerm_walks})}
    \\
    & non-strict 
        & \NP-c \szTCref{(\Cref{thm:SD_TPD})}
        & \NP-c \szTCref{(\Cref{thm:TTD_para})}
        & \NP-c \szTCref{(\Cref{thm:NSUD_TWD})}
        & \PP\ \szTCref{(\Cref{thm:UD_NS_TWD_fixed_terminals})}
\end{tblr}
\caption{Overview of our results on general graphs for $k\geq3$ with the corresponding references.}
\label{tab:overview_table}
\end{table*}

\section{Introduction}
\label{sec:intro}

Networks are fundamentally designed to enable the transportation of various ``entities'', ranging from people and physical goods to data and information. Achieving efficient transportation requires the deployment of vehicles and connections, which most often are of limited quantity. Consequently, optimizing the use of these resources presents a natural challenge in network planning and design.
Take, for example, a train company that decided on optimal train connections to meet public demand. Next, the company has to decide on the number of trains to be deployed on a daily schedule to realize these connections effectively in order to maximize profit. This highlights the critical issue of how to allocate limited resources efficiently to meet operational goals while balancing cost and demand.

This issue is even more apparent in networks that are operated by multiple companies rather than a single authority.
%
In such scenarios, each company independently sets its own desired connection times between the depots based on various and individual criteria like public service needs, profits, regulations and resource restrictions. Their main objective is to optimally ensure their own connections at the designated time.
However, such decentralized planning often leads to suboptimal use of resources. This issue has been identified by the Supply Chain Management and Logistics industries, which have observed that ``{\em while logistics improvements may be difficult for each individual customer, they can be realized with collaborative operation among multiple customers}''\footnote{\url{https://www.logisteed.com/en/3pl/joint/}}.

To illustrate, consider a scenario where company X wants to operate a line between points A and B at 08:00, and from point B to A at 20:00, where each connection takes one hour. 
Company Y wants to operate a line between points B and C at 09:00, and from point C to B at 17:00, also with each connection taking one hour. 
Independently, both companies require one vehicle each to operate their respective lines.
However, if they collaborated, they could share a vehicle for both lines with a single {\em journey}: from A to B at 08:00, from B to C at 09:00, from C to B at 17:00, and from B to A at 20:00.
\begin{figure}[h]
    \centering
    \includegraphics[width=0.6\linewidth]{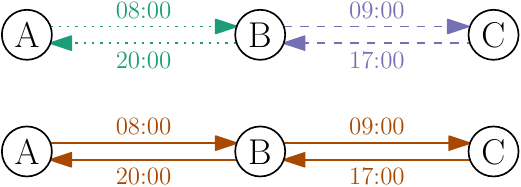}
    \label{fig:intro-example}
\end{figure}

It is apparent that 
optimizing resources on networks is important across various fields, with each scenario presenting unique complexities.
This poses a dilemma: should we create a specific model for each different scenario and study it independently, or develop a generalized, possibly simplified model to study many scenarios at once? 
Applied research typically chooses the former, tailoring models to individual applications, while theoretical sciences favor the latter, aiming for generalized models with broader applicability.

In this paper, we follow route 2 and provide a general model that is both simple and realistic enough to represent multiple scenarios. 
We want to study the problem of finding the minimum number of resources needed to realise the connections of a network.
To account for the dynamic nature of the networks with time-specific connections, 
we use the prominent framework of the \emph{temporal graphs model} \cite{kempe02-temporal, Michail2015}.
A temporal graph consists of a static set of vertices and a dynamic set of edges, called temporal edges, with a labelling function indicating at which time steps each edge is available.
A trip through a temporal graph, known as a {\em temporal journey}, has to respect these time constraints. 

\subsection*{Our Contribution}
\footnotetext{For temporal graphs with unique labels at each vertex.}
Our contribution is twofold: (a) formalize the planning problem highlighted in the example above as an optimization problem on temporal graphs; (b) provide the landscape of its computational complexity for a rich variety of settings. 
\subsubsection*{The Exact Edge-Cover.}
We study the following problem. %
\begin{quote}
{\em ``Given a temporal graph and a natural number $k$, can we cover every temporal edge of the graph, exactly once, with at most $k$ temporal journeys?''} 
\end{quote}
\noindent
Motivated by the variety of real-life scenarios that our model can capture, we perform a thorough study of our problem with respect to four dimensions.
First, the type of temporal journey, which can be a {\em temporal path}, {\em trail}, or {\em walk}. This distinction is important as, for example, trains usually should not revisit a station in a single trip, which can only be modeled by a temporal path.
Second, the nature of the temporal journey, which can be {\em strict}, or {\em non-strict}. In a strict journey, at most one temporal edge can be used per time step. 
For example, in online networks, information can traverse several connections immediately due to its inherent transmission speed.
Third, the type of the graph, which can be either {\em directed} or {\em undirected}. Physical networks are usually directed, whereas information networks allow for bidirected connections.
Fourth, having {\em fixed} or {\em unfixed terminals} for the journeys, i.e., whether the starting and ending points of the journeys are given or not.
This captures that depots might be available only at specific locations in the network.
A summary of our technical results is depicted at \Cref{tab:overview_table}. 

\subsubsection{Our Technical Results.}
The first part of our technical results focuses on the case where the terminals are not fixed. As a warm-up, we observe that when $k=1$, the problem resembles Eulerian journeys, and provide a linear time algorithm for all our settings.
We extend our positive results to $k=2$, where we show that {\em almost} all versions of our problem can be reduced to a \sat{2} problem, enabling a polynomial-time solution. However, this approach fails for non-strict trails, in which the problem is \NP-complete.


For $k \geq 3$, 
positive results become scarce. 
In \Cref{sec:pathsAndtrails}, we study paths and trails, proving intractability for all versions of the problem for any constant $k$. Furthermore, we show that no version can admit an $\alpha$-approximation for any constant $\alpha$ unless $\PP=\NP$.  
On the other hand, walks (\Cref{sec:walks}) offer some positive results. We observe that when the number of walks $k$ is constant, all versions of our problem can be solved efficiently.
For arbitrary $k$, we additionally derive an efficient algorithm for the case of strict walks on directed graphs. Unfortunately, we complement the above by showing that every other version is \NP-complete.


The second part of our technical results (\Cref{sec:fxdTerminals}) focuses on fixed terminals, where the start and end points of the journeys are given. 
We begin by observing for paths and trails our hardness reductions hold even with fixed terminals.
For walks, the situation is different. 
We provide a dynamic program for non-strict walks on directed graphs.
On undirected graphs, the problem resembles a circulation flow.
To capture the temporal aspect of the flow, we require a complex construction, making this the most technically involved result of the paper. The resulting algorithm solves the problem for non-strict walks.
For strict walks, the algorithm works only when the labels at each vertex are distinct. Additionally, we provide a dynamic program for graphs where there are exactly $k$ or 0 edges per time step.

\subsection*{Related Work}
The minimum number of resources needed to realize a network has been studied in a lot of different contexts.
For example, the \textit{minimum fleet size problem} asks for the minimum number of vehicles needed to cover every trip in a network.
Unlike our model, this problem is studied only on directed networks that can be modeled to contain no cycles and aims to cover every node rather than every edge \cite{vazifeh2018addressing}.
Another example is the \textit{rolling stock problem}, which asks for the minimum number of wagons needed to meet the passenger demand of a train network \cite{alfieri2006efficient}.
Here, the train lines are given together with a capacity demand on each connection that needs to be met.
The goal is to optimize the number of cars needed by leveraging car exchanges between changes on stations.
In contrast, our model is given individual connections (each with unit capacity) and aims to find the optimal train lines. 

The study of edge covering problems on static graphs was initiated by \citet{ErdösTheRepresentation1966} and \citet{LovaszOnCovering1968}.
This led to an extensive line of work studying the minimum number of gadgets needed to cover the edges of a graph. The usual gadgets considered are cliques, cycles and paths. As far as paths are concerned, most of the research focuses on proving existential upper bounds on the number of gadgets. 
\citet{ChungOnTheCoverings1980} conjectured that a graph can be covered with at most $\ceil{\frac{n}{2}}$ paths, which was settled by \citet{FanSubgraph2002}. For an extensive overview of the covering problem on static graphs, we refer to this survey \cite{SchwartzAnOverview2022}.

In the case of temporal graphs, many pieces of work have studied reachability problems adjacent to ours \cite{zschoche:LIPIcs.STACS.2023.55,BiloSparse2022,BiloBlackout2022,DBLP:journals/jcss/EnrightMS21,DBLP:journals/ans/BentertHNN20,DBLP:journals/ijfcs/XuanFJ03,DBLP:conf/atal/DeligkasEGS23,DBLP:conf/ijcai/DeligkasES23}.
The most prominent one is the temporal exploration problem, where the goal is to compute an earliest-arrival walk that visits every vertex of the graph.
It was introduced by \citet{ErlebachOnTemporal2021} and has been extensively studied since \cite{ArrighiKernelizing2023,ErlebachPrameterised2023,AdamsonFaster2022}.

Recently, \citet{BumpusEdgeExploration2023}, and \citet{MarinoEulerian2023} introduced an edge variant of the problem, where the goal is to decide whether the static edges of a temporal graph can be covered by an Eulerian circuit or trail.
The main difference between their work and ours is that we want to cover every temporal edge, whereas they focus on covering every static edge.
Additionally, they consider covering the edges with one trail, while we study an arbitrary number and different types of journeys. 
We highlight that the two models are orthogonal to one another and no result can be transferred between them.


%

\section{Preliminaries}
    \label{sec:prelims}
    For $n\in \mathbb{N}$, we denote $[n]:= \{1,2, \ldots, n\}$.
    A {\em temporal graph} $\gcal := \tuple{G, \ecal}$ is defined by an {\em underlying} static graph  $G=(V,E)$ and a sequence of edge-sets $\ecal = (E_1, E_2, \ldots, E_{\tmax})$ with $E = E_1\cup E_2\cup \dots \cup E_{t_{max}}$.
    We refer to $E$ as the set of \emph{static} edges of \gcal and to $\ecal$ as the set of \emph{temporal} edges. 
    The {\em lifetime} of $\gcal$ is \tmax.
    An edge $e \in E$ \emph{has label} $i$ if $e \in E_i$.
    We denote directed edges with \tdirect{u}{v}{t} and undirected edges with \tundirect{u}{v}{t}.
    The {\em edgestream representation} $\edgestream$ of \gcal is an ordered list of all temporal edges, sorted by increasing time label. 
    By a {\em snapshot} $G_t$ of $\gcal$ we refer to the subgraph $(V, E_t)$ containing only the edges available at time $t \in [\tmax]$.
%
    The temporal degree $\tdegree(v)$ of a vertex $v$ denotes the number of temporal edges adjacent to $v$.
    On directed graphs, we further define $\tindegree(v)$ and $\toutdegree(v)$ as the number of incoming and outgoing temporal edges of $v$, respectively.

    \paragraph{Temporal Journeys.}
    A {\em temporal journey} $J$ in $\gcal=\tuple{(V,E), \ecal}$
    is a sequence of adjacent temporal edges that respect time, connecting a {\em start-terminal} with an {\em end-terminal}. 
    Formally, if $J = ((e_1,t_1),\dots,(e_{\ell},t_{\ell}))$ is a temporal journey, then for every 
    $i \in [\ell]$ holds that: $e_i \in E_{t_i}$; $e_i$ is adjacent to $e_{i+1}$; and $t_i \leq t_{i+1}$.
    If we require that $t_{i} < t_{i+1}$ for all $i\in[\ell]$, we call the journey {\em strict}, otherwise we call it {\em non-strict}.
    For a temporal journey $J$, we denote by $J[t_1,t_2]$ the {\em sub-journey} between $t_1$ and $t_2$, i.e., the sub-sequence of temporal edges $(e_i,t_i)$ of $J$, such that $t_i\in [t_1,t_2]$.
    If the temporal graph is directed, then we have {\em directed} journeys that have to respect the direction of the edges, otherwise, we have {\em undirected} journeys. We focus on three different types of temporal journeys:
    \begin{itemize}
        \item {\em temporal walks}, which do not impose any extra constraints;
        \item {\em temporal trails}, where no static edge is visited more than once by the journey;
        \item {\em temporal paths}, where no vertex is visited more than once by the journey.
    \end{itemize}

    \paragraph{Parameterized complexity.}
    We refer to the standard books for a basic overview of parameterized complexity theory~\cite{CyganFKLMPPS15,DowneyFellows13}. At a high level, parameterized complexity studies the complexity of a problem with respect to its input size, $n$,  and the size of a parameter $k$. A problem is {\em fixed parameter tractable} by $k$, if it can be solved in time $f(k)\cdot \poly(n)$, where $f$ is a computable function. 
    \iflong A less favorable, but still positive, outcome is an $\XP{}$ \emph{algorithm}, which has running-time $\bigoh(n^{f(k)})$; problems admitting such algorithms belong to the class $\XP$. \fi
    Showing that a problem is $\Wtwo$-hard rules out the existence of a fixed-parameter algorithm under the well-established assumption that $\Wtwo\neq \FPT$.
    
\section{Exact Edge-Cover}
    Exact edge-covers by paths in static graphs have been studied by \citet{donald1980upper, bolyaicovering, pyber1996covering}, where this structure is called an \textit{edge-disjoint path cover}.
    \iflong
    We define the generalisation of this problem to different types of journeys.
    \begin{definition}[Static Exact Edge-Cover by Journeys]
        Let $G=(V,E)$ be a static graph and $\Phi\in\{\text{path, walk, trail}\}$.
        A collection $\jcal=\{J_1,\dots,J_k\}$ of journeys of type $\Phi$, is called an {\em exact edge-cover (\eecShort) by $\Phi$} for $G$ if for all edges $e\in E$ there is exactly one $J_i\in \jcal$ with $e\in J_i$.
    \end{definition}
    The temporal analogue is then defined accordingly.
    \begin{definition}[Temporal Exact Edge-Cover by Temporal Journeys]
        Let $\gcal=\tuple{(V,E), \ecal}$ be a temporal graph and $\Phi\in\{\text{path, walk, trail}\}$.
        A collection $\jcal=\{J_1,\dots,J_k\}$ of journeys of type $\Phi$, is called a {\em temporal exact edge-cover by $\Phi$} for \gcal if for all temporal edges $e\in \ecal$ there is exactly one $J_i\in \jcal$ with $e\in J_i$.
    \end{definition}
    \fi
    \ifshort
    We define the temporal generalization on different types of journeys.
    \begin{definition}[Temporal Exact Edge-Cover by Temporal Journeys]
        Let $\gcal=\tuple{(V,E), \ecal}$ be a temporal graph and $\Phi\in\{\text{path, walk, trail}\}$.
        A collection $\jcal=\{J_1,\dots,J_k\}$ of journeys of type $\Phi$, is called a {\em temporal exact edge-cover (eec) by $\Phi$} for \gcal if for all temporal edges $e\in \ecal$ there is exactly one $J_i\in \jcal$ with $e\in J_i$.
    \end{definition}
    \fi

    We study the problem of finding temporal exact edge-covers of size $k$ by the (non-)strict version of each journey type in (un)directed temporal graphs.
    For $\Phi\in\{\textsc{path}, \textsc{trail}, \textsc{walk}\}$, we define \TJEECnoK{(N)S-} as follows.
\noindent
    \begin{center} 
    \begin{tcolorbox}[colback=VanillaIce,colframe=MexicanRed!75!white, title=~\textsc{(Non-) Strict Temporal $\Phi$ Exact Edge-Cover},left=-1.5mm,top=0mm,bottom=0mm,right=0mm,boxsep=1mm]
        \begin{tabular}{p{.15\textwidth}p{.8\textwidth}}
            \textbf{Input:} & \parbox[t]{.74\textwidth}{A temporal graph $\gcal$ and $k\in \mathbb{N}$.}\\
            \textbf{Question:} & \parbox[t]{.74\textwidth}{Does there exist a temporal exact edge-cover of (non-)strict $\Phi$ for \gcal of size $k$?}\\
        \end{tabular}
      \end{tcolorbox}
      \end{center} 
    Checking whether a given collection of journeys is an exact edge-cover can be done in polynomial time by verifying that each temporal edge appears exactly once and that all journeys are of type $\Phi$. Thus, \TJEECnoK{(N)S-} is in the class \NP.
    
    \noindent
    \iflong
    We refer to the static analogue of finding an exact edge-cover of $\Phi$ in a static graph $G$ as \textsc{S$\Phi$EEC}. Note that \SWEC is a natural generalisation of finding an Euler walk in a static graph, which corresponds to an eec of size $k=1$ and is well known to be computable in linear time \cite{fleischner1990eulerian}.
    \fi
    \subsection{One Journey -- Extending Eulerian Journeys} \label{subsec:one-journey}
    As a warm-up, we show that when $k=1$, all versions of our problem can be solved in polynomial time. While our algorithm follows the same strategy for every journey type, the approach needs to differentiate depending on whether the journeys are strict or not.
    \ifshort
    \begin{proposition}[$\star$]
        For $\Phi\in\{\textsc{path}, \textsc{trail}, \textsc{walk}\}$, we can solve \TJEECwithK{1}{(N)S-} in polynomial time on both directed and undirected graphs.
    \end{proposition}
    For strict journeys, we attach the edges in temporal order and check whether this forms the desired journey. For non-strict journeys, we additionally observe that a set of edges appearing at the same time has to form an Eulerian journey which can be found in polynomial time \cite{fleischner1990eulerian}. 
    \fi

    \iflong
    For strict temporal paths on directed graphs, we can find an optimal \eecShort by starting with the unique earliest edge and greedily attaching all other edges in temporal order. 
    \begin{observation}
        \TPEECwithK{1}{S-} on directed graphs can be solved in linear time.
    \end{observation}
    For the other variants, we can extend this approach.
    First, we analyse the case of strict journeys in undirected graphs. Observe that \gcal can have at most one temporal edge per time step.
    Because of this, the start-terminal in a directed graph is uniquely defined.
    For an undirected \gcal, we can guess which of the two endpoints of the earliest edge is the start-terminal.
    Now, we proceed as before by greedily checking if the next time edge can be attached to the current prefix journey. 
    If some edge is not connected to the current endpoint or attaching it would break a constraint of the journey type, we can safely reject the guess.
    \begin{lemma}
        For $\Phi\in\{\textsc{path}, \textsc{trail}, \textsc{walk}\}$, we can solve \TJEECnoK{S-} for $k=1$ in directed and undirected graphs in $\mathcal{O}(\lvert \ecal\rvert + \lvert V\rvert)$ time.
    \end{lemma}
    
    In the non-strict case, we additionally have to deal with sets of edges appearing at the same time step.
    In each snapshot, the journey needs to form either an Eulerian path (for paths) or an Eulerian trail (for trails and walks).
    If we are looking for a path, we can identify the two vertices at which the path has to start or end in the particular snapshot.
    If we are looking for a walk or trail, the snapshot can contain a cycle: If there are two vertices of odd degree, the journey has to start or end at those in the particular snapshot.
    If all vertices have even degree, the Eulerian trail has to start and end in the same, but arbitrary, vertex.
    If there is a first snapshot where the journey needs to travel between specific two vertices, we guess which of them will be the start and continue as in the strict case. 
    If no such snapshot exists, we check whether each snapshot induces a connected graph and compute the intersection of the vertex sets with non-zero degree at each snapshot. 
    \begin{lemma}
    \label{lem:one_journey}
        For $\Phi\in\{\textsc{path}, \textsc{trail}, \textsc{walk}\}$, we can solve \TJEECnoK{NS-} for $k=1$ in directed and undirected graphs in $\mathcal{O}(\lvert \ecal\rvert + \lvert V\rvert)$ time.
    \end{lemma}
    \fi
    
\ifshort
\begin{figure*}[t]
        \centering
        \includegraphics[width=0.7\textwidth]{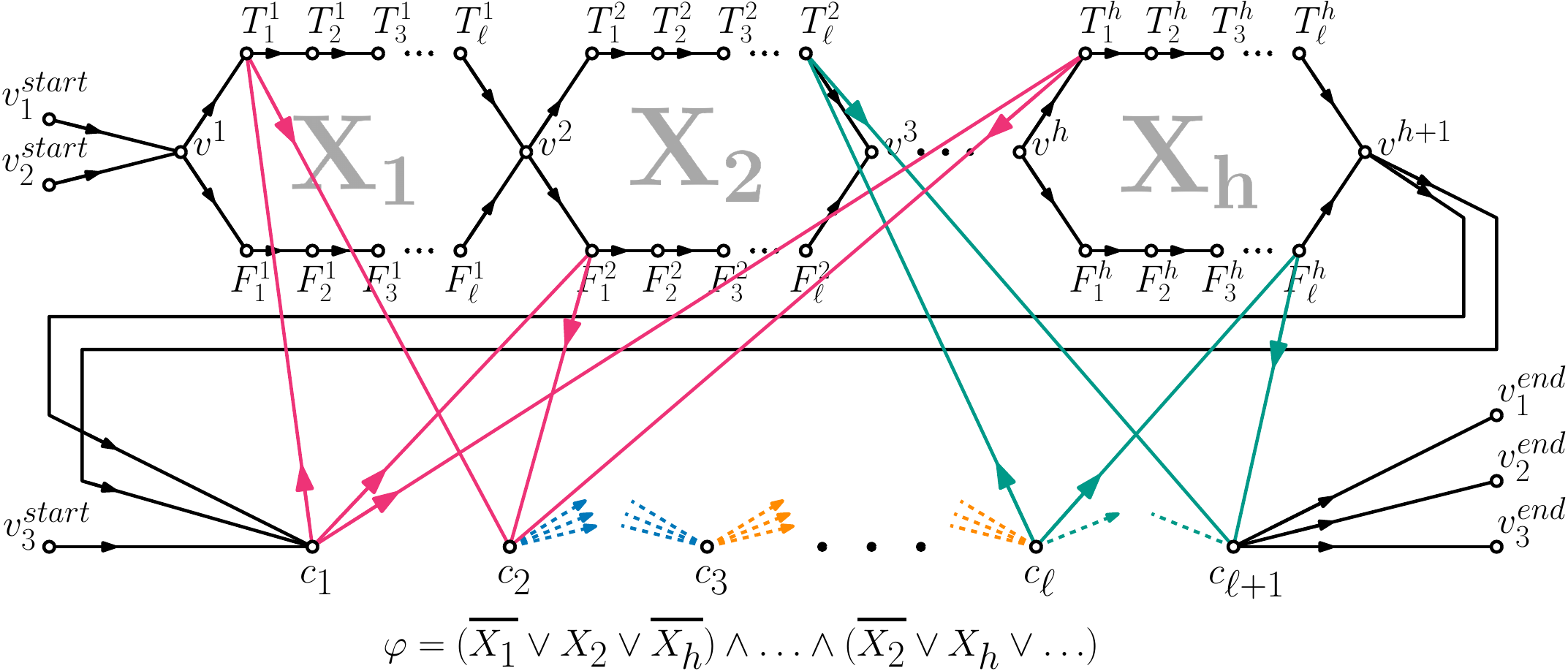}
        \caption{Illustration of the construction for the \NP-hardness reduction in \Cref{thm:SD_TPD} for $k=3$. The two edges from $v^{h+1}$ to $c_1$ connecting the assignment-gadget on the top and satisfaction-gadget on the bottom are drawn as two separate edges. With this, every drawn edge has exactly one time label increasing with the direction of the edges.
        }
        \label{fig:paraNPhard}
    \end{figure*}
\fi
\subsection{Two Journeys -- \sat{2} Approach}\label{subsec:two_journeys} 
    While \TPEECwithK{2}{(N)S-}, \TTEECwithK{2}{S-} and \TWEECwithK{2}{(N)S-} are solvable in polynomial time for both directed and undirected graphs, we show that \TTEECwithK{2}{NS-} is \NP-hard on directed and undirected graphs. 

    \ifshort Towards the positive results, we utilize a polynomial-time algorithm for \sat{2}. \fi 
    \iflong 
    For walks, we observe in a later section (\Cref{sec:walks}) that \TWEECnoK{(N)S-} can be solved in polynomial time when the size of the eec we are looking for is fixed. This immediately implies the following for walk eecs of size 2.
    \begin{corollary}\label{lem:2TWD}
        \TWEECwithK{2}{(N)S-} can be solved in polynomial time.
    \end{corollary}
    
    Towards the positive results for paths and trails, we utilize a polynomial time algorithm for \xorsat{2}.
    We provide a detailed description for strict paths in directed graphs in \Cref{lem:D_S_2TPD} and extend it to the other path variants in \Cref{lem:UD_S_2TPD,lem:D_NS_2TPD,lem:UD_NS_2TPD}. Using the same approach, we prove the claim for strict trails on directed and undirected graphs in \Cref{lem:D_S_2TTD,lem:UD_S_2TTD}, respectively.

    \subsubsection{Two Paths}
    We begin with the main idea for the construction. 
    \fi
    Intuitively, as paths cannot revisit vertices we know that on directed temporal graphs every vertex can have at most two incoming and two outgoing edges
    , each \iflong of which has to be \fi covered by exactly one path.
    \iflong Obviously, the earliest temporal edge incident to a vertex $v$ that is not one of the terminals has to be an incoming edge. \fi
    Now, if the temporal edges incident to a vertex are alternatingly incoming and outgoing [\textit{in-out-in-out}], we \textbf{have to} match the first \textit{in}-edge with the first \textit{out}-edge.
    However, when both incoming edges appear before the outgoing edges [\textit{in-in-out-out}], we \textbf{can choose} which \textit{in} is matched with which \textit{out}. 
    \iflong We refer to the latter also as a \textit{hub}. \fi
    This can be formalised via a \sat{2} formula with one variable for each vertex with \textit{in-in-out-out} edge appearance and one clause for each vertex with \textit{in-out-in-out} edge appearance. Setting a variable to True corresponds to the first path taking the earliest \textit{out}-edge at that vertex and setting \iflong the variable \fi \ifshort it \fi to False corresponds to the first path taking the latest \textit{out}-edge.
    \ifshort
    \begin{theorem}[$\star$]
    \label{thm:D_S_2TPD}\label{lem:D_S_2TPD}
        \TPEECwithK{2}{(N)S-}, \TTEECwithK{2}{S-} and \TWEECwithK{2}{(N)S-} 
        can be solved in polynomial time on directed and undirected 
        temporal graphs. 
    \end{theorem}
    \fi
    \iflong
\newcommand{\visited}{\text{visited}}
\newcommand{\sOne}{\ensuremath{s_1}\xspace}
\newcommand{\sTwo}{\ensuremath{s_2}\xspace}
\newcommand{\POne}{\ensuremath{P_1}\xspace}
\newcommand{\PTwo}{\ensuremath{P_2}\xspace}
\newcommand{\vOne}{\ensuremath{v_1}\xspace}
\newcommand{\vTwo}{\ensuremath{v_2}\xspace}
\begin{lemma}[Strict Directed Paths]
\label{thm:D_S_2TPD}\label{lem:D_S_2TPD}
    For $k=2$, we can solve \TPEECnoK{S-} for directed temporal graphs in $\mathcal{O}(\lvert\ecal\rvert + \lvert V\rvert)$ time.
\end{lemma}
\begin{proof}
    We reduce the problem to \xorsat{2}.
    Let $\gcal=\tuple{(V,E), \ecal}$ be a temporal graph and assume that one path cannot cover $\gcal$.
    Without loss of generality each vertex has temporal in- and out-degree of at most two. 

    \vspace{0.4em}\noindent\textit{Start vertices.}\quad Every temporal path has exactly one start-vertex. As there are $n$ vertices and 2 paths, there are $\binom{n}{2}$ possibilities which can be iterated through in polynomial time.
    From hereon, we assume the two start-vertices \sOne\ and \sTwo\ to be given and \sOne to be adjacent to the earlier edge.


    
    \vspace{0.4em}\noindent\textit{Path building.}\quad To construct the \xorsat{2} formula $\varphi$, we use a dynamic program to  track the two endpoints of the current paths at each time step, while storing visited vertices.
    
    We initiate $X_0$ as the current variable and 
    start the two paths \POne and \PTwo at \sOne and \sTwo, respectively, by storing the current endpoint of $P_i$ as $v_i=s_i$. We also set the currently visited vertices of both paths to $V^0_1=\emptyset=V^0_2$. 
    
    We process the edges $e=((u,v),t)$ in temporal order. Let $X_\alpha$ be the current variable.
    As long as $\vOne\neq\vTwo$, we extend $P_i$ if $u=v_i$. If neither paths can be extended by $e$, the algorithm terminates and returns that there is no path eec of size 2.
    If $P_i$ is extended by $((u,v),t)$, we mark $v$ as visited and add it to $V^\alpha_i$ -- the visited vertices of $P_i$ in phase $\alpha$.
    Next, we check if $v$ has been visited before. If $v$ has been visited by $P_j$ in phase $\gamma\leq\alpha$, let $Y=X_\gamma$ if $j=1$ and $Y=\overline{X_\gamma}$ if $j=2$. Similarly, let $Z=X_\alpha$ if $i=1$ and $Z=\overline{X_\alpha}$ if $i=2$.
    Now, we add the clause $Y\oplus Z$.
    This clause ensure that either $P_1$ takes the early edge at the hub-vertex $\gamma$ and $P_2$ takes the early edge at the hub-vertex $\alpha$ or vice versa.
    
    Whenever the two paths ``meet'' in a hub-vertex at time step $t$ ($v_1=v_2$), we generate a new variable $X_{\alpha+1}$ and store it as the current variable.
    We reset the two paths by setting $V^{\alpha+1}_1=\emptyset=V^{\alpha+1}_2$.
    The next edge has to start at $v_1$ and, without loss of generality, we attach it to the earlier path $P_1$.

    If every edge is processed without returning that there is no path eec of size 2, we have constructed a (possibly empty) \xorsat{2} formula $\varphi$. We show that this formula is satisfiable if and only if $\gcal$ has a path eec of size 2.

    \bigparagraph{($\Rightarrow$)\quad } Let $\beta$ be a satisfying assignment of $\varphi$, $a=\lvert var(\varphi)\rvert$ variables. Let $h_i\in V$ for $i\in[a]$ denote the [\textit{in-in-out-out}] hub-vertices corresponding to each variable. 

    We start the two edges in the corresponding start-vertices. We  attach each time edge in chronological order to the unique adjacent temporal path until the two paths meet in a hub-vertex $h_i$.
    Now we use the assignment $\beta$ to decide which of the paths exits $h_i$ first: \POne leaves the vertex first if and only if $\beta(X_i)=1$.
    Since $\beta$ is a satisfying assignment, this ensures that no vertex is visited by the same path twice.
    
    \bigparagraph{($\Leftarrow$) \quad} 
    Let $\mathcal{P}$ be a path exact edge-cover of size 2 for \gcal and let $\varphi$ be the constructed formula.
    We construct a satisfying assignment $\beta$ by setting $\beta(X)=1$ if and only if \POne exits the $i$\textsuperscript{th} hub-vertex first. Since $\mathcal{P}$ is a path eec, the two contained paths do not revisit any vertices and therefore, every clause of $\varphi$ is satisfied.

    \vspace{0.4em}\noindent\textit{Polynomial runtime.}\quad 
    The construction of the \xorsat{2} formula from \ecal takes $\mathcal{O}(\lvert\ecal\rvert + \lvert V\rvert)$ time and the formula has $x=\mathcal{O}(\lvert V\rvert)$ variables and $y=\mathcal{O}(\lvert V\rvert)$ clauses.
    An \xorsat{2} formula with $x$ variables and $y$ clauses can be transformed into a \sat{2} formula with $x'=x$ variables and $y'=2y$ clauses by the following equivalence $(L_1\oplus L_2)\equiv(L_1\vee L_2)\wedge(\overline{L_1} \vee \overline{L_2})$.
    A satisfying assignment for such a \sat{2} formula can be found in $\mathcal{O}(x'+y')$ time if it exists \cite{even1975complexity}.
    The overall running time of the algorithm is thus $\mathcal{O}(\lvert\ecal\rvert + \lvert V\rvert)$.
\end{proof}

    This \xorsat{2} approach can be extended to undirected graphs and non-strict paths using arguments similar to the extensions for $k=1$ in \Cref{subsec:one-journey}.
    \begin{figure*}[h]
        \centering
        \includegraphics[width=\textwidth]{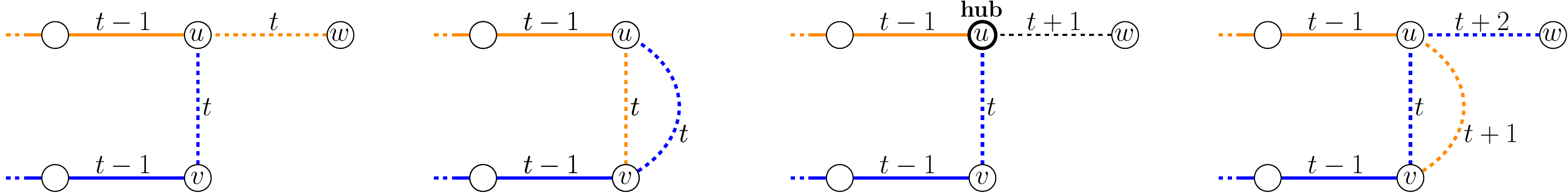}
        \caption{Illustration of the cases for possible adjacent edges for two journeys. The cases are the same for all three types of journeys, only the behaviour of the journeys may change.}
        \label{fig:ktwo-trails-cases}
    \end{figure*}
    \begin{lemma}[Non-Strict Directed  Paths] \label{lem:D_NS_2TPD}
        For $k=2$, we can solve \TPEECnoK{NS-} for directed temporal graphs in $\mathcal{O}(\lvert\ecal\rvert + \lvert V\rvert)$ time.
    \end{lemma}
    \begin{proof}
        We extend the argument of \Cref{thm:D_S_2TPD} from strict  to non-strict paths. As the conditions on revisiting vertices do not change, neither do the restrictions on the in- and out-degree of every vertex nor the properties of the start-vertices.

        Non-strict paths can take edges with the same time label consecutively. Therefore, we add a local check to the process of \Cref{thm:D_S_2TPD} that is used only when we encounter a time label $t$ with $\lvert E_t\rvert>1$.
        
        If one of the current endpoints $v_1$ has two outgoing edges with time label $t$, there has to be an incoming path from the other current endpoint $v_2$ at this time label $t$. If that is not the case, one of the two outgoing edges from $v_1$ cannot be covered. This path can be found in linear time.

        Therefore, assume $v_1=v_2=v$, \ie the two current paths meet in a hub-vertex. That means the current variable is a new variable, say $X^\alpha$. 
        If $out(v)\cap E_t = 0$, the edges cannot be attachted to the current paths 
        and an eec with two paths is not possible.
        If $out(v)\cap E_t = 1$, the edges in $E_t$ have to form an Eulerian path starting in $v$ covered by one path. We can compute this in linear time, add all the vertices to $V^\alpha_1$ and check whether we have to add clauses to $\varphi$ as in the original process.
        Lastly, if $out(v)\cap E_t = 2$, both paths will have to exit the vertex at time step $t$. We proceed with a breadth-first search on the edges of $E_t$ starting in $v$. Out of the two outgoing edges, we choose one at random as the \emph{earlier} edge and proceed with attaching the outgoing edges if possible, until either all edges are attached, or the two paths meet again in another vertex. We then repeat the case-distinction on $out(v)\cap E_t$.
        After processing all edges of $E_t$, we continue with the next time edge(s).        

        The local check does not add any more processing time to the overall process which is thus still $\mathcal{O}(\lvert\ecal\rvert + \lvert V\rvert)$.
        \end{proof}   
    Moving on to strict paths in undirected graphs, we have to add a process to process undirected edges. We provide a case distinction based on where an undirected edge attaches.

    \begin{lemma}[Strict Undirected Paths] \label{lem:UD_S_2TPD}
        For $k=2$, we can solve \TPEECnoK{S-} for undirected temporal graphs in $\mathcal{O}(\lvert\ecal\rvert + \lvert V\rvert)$ time.
    \end{lemma}
    \begin{proof}
        Observe that the direction of the first temporal edge is clear for both prefixes from the given start-terminals. Assume we processed all temporal edges until time step $t-1$ and let $v_1^{t-1}$, $v_2^{t-1}$ denote the endpoints of the current prefixes.
        We make a case distinction based on how the next undirected edge $e_i=\tundirect{u}{v}{t}$ can be connected to the prefix paths.
        In the following, we omit mirrored cases like $u=v_1^t$, $v=v_2^t$ and $u=v_2^t$, $v=v_1^t$.
        
        \vspace{0.4em}\noindent\textit{Case 1.}\quad 
            The next edge is attachable to exactly one prefix, \ie $v_i^{t-1}=u\neq v_2^{t-1}$ and $v_1^{t-1}\neq v\neq v_2^{t-1}$. The edge has to be attached to that prefix.
        
        \vspace{0.4em}\noindent\textit{Case 2.}\quad 
            The next edge is attachable to both prefixes and $v_1^t=v_2^t$, \ie the prefixes meet in a hub-vertex. In this case, we proceed in the same way as in the directed case (\Cref{lem:D_S_2TPD}).
        
        \vspace{0.4em}\noindent\textit{Case 3.}\quad
            The next edge is attachable to both prefixes and $v_1^t\neq v_2^t$, \ie $v_1^{t-1}=u\neq v_2^{t-1}$ and $v_1^{t-1}\neq v= v_2^{t-1}$.
            In this case it is unclear which prefix should be extended and we check the next edge $e_{i+1}=\tundirect{u'}{v'}{t'}$ in the edgestream.
            We distinguish four subcases.
            Refer to \Cref{fig:ktwo-trails-cases} for an illustration.

            \vspace{0.2em}\hspace{0.4em}\noindent\textit{Case 3(a).}\quad 
            Let $\{u,v\}=\{u',v'\}$ and $t=t'$.
            The two prefixes have to take one edge each and their endpoints are switched after time $t$, \ie $v_1^{t}=v_2^{t-1}$ and $v_2^{t}=v_1^{t-1}$.
            
            \vspace{0.2em}\hspace{0.4em}\noindent\textit{Case 3(b).}\quad 
            Let $\{u,v\}=\{u',v'\}$ and $t<t'$.
            Again, the two prefixes have to take one edge each and their endpoints are switched after time $t+1$.
            However, if the second next edge $e_{i+2}$ after $e_{i+1}$ is also at time step $t+1$ and attaches at $u$, then prefix two has to take the edge $e_i$ at time step $t$ in order to be able to take $e_{i+2}$. Vice versa, if $e_{i+2}$ attaches at $v$, prefix one has to take $e_i$ at time $t$.

            \vspace{0.2em}\hspace{0.4em}\noindent\textit{Case 3(c).}\quad 
            Let $\{u,v\}\neq\{u',v'\}$ and $t=t'$.
            Then wlog $u'=u=v_1^{t-1}, v'=w\neq v$ and prefix one has to be extended to $w$ while prefix two is extended to $u$.
            After $t$, we have $v_1^{t}=w$ and $v_2^{t}=u=v_1^{t-1}$.
            
            \vspace{0.2em}\hspace{0.4em}\noindent\textit{Case 3(d).}\quad 
            Let $\{u,v\}\neq\{u',v'\}$ and $t<t'$.
            In this case prefix two has to take $e_i$ and the two prefixes meet in $u$ after time step $t$ which is therefore a hub-vertex. From the hub-vertex, we continue as in Case 2.
        
        These additional checks can be done in linear time, which yields an overall runtime of $\mathcal{O}(\lvert\ecal\rvert + \lvert V\rvert)$.    
    \end{proof}   
    
    Lastly, we combine the arguments 
    of \Cref{lem:D_NS_2TPD} and \Cref{lem:UD_S_2TPD} to get an algorithm for non-strict paths in undirected graphs.
    Meaning, for every set of edges $E_t$ with the same time label $t$, we add a local check starting in the current endpoints $v_1^{t-1}$ and $v_2^{t-1}$ and whenever we encounter an undirected edge $\tundirect{u}{v}{t}$ attachable to both paths, we look ahead at which end of $(u,v)$ the next time edge will appear.
    \begin{lemma}[Non-Strict Undirected Paths] \label{lem:UD_NS_2TPD}
        For $k=2$, we can solve \TPEECnoK{NS-} for undirected temporal graphs in $\mathcal{O}(\lvert\ecal\rvert + \lvert V\rvert)$ time.
    \end{lemma}

\subsubsection{Two Strict Trails}
Finally, we discuss the strict trails in directed and undirected graphs; once again using the \xorsat{2} approach.
However, instead of creating a clause for every [\textit{in-out-in-out}] vertex, we create a clause for every static edge $(u,v)$ which has two time labels (and has to be visited by both trails).
We will again create a new variable for every hub-vertex, \ie every vertex at which the trails meet in one time step. Note that those do not correspond to [\textit{in-in-out-out}] vertices anymore as trails can revisit vertices and we can have vertices with degree higher than two. Otherwise, the algorithm will work as for paths.
We first prove the result for strict trails in directed graphs.
\begin{lemma}[Strict Directed Trails] \label{lem:D_S_2TTD}
    For $k=2$, we can solve \TTEECnoK{S-} for directed temporal graphs in $\mathcal{O}(\lvert\ecal\rvert + \lvert V\rvert)$ time.
\end{lemma}
\begin{proof}
    We initiate variable $X^0$ as the current variable.
    As for paths, we start an early trail $T_1$ and a late trail $T_2$ at the first two edges and greedily attach edges to the two current trails until they meet in a hub-vertex. During this process, we store in $E_i^\alpha$ the \textbf{edges} visited by trail $T_i$ in phase $\alpha$.

    Rest of process same as \Cref{lem:D_S_2TPD} and the same arguments hold when one substitutes the property check \textit{no vertex is revisited by the same journey} by \textit{no edge is revisited by the same journey}.

    The constructed \xorsat{2} formula constructed from this process contains $x=\mathcal{O}(\lvert V\rvert)$ variables and $y=\mathcal{O}(\lvert E\rvert)$ clauses and so the overall running time is still $\mathcal{O}(\lvert\ecal\rvert + \lvert V\rvert)$.
\end{proof}

For trails in undirected graphs, we copy the case distinction of \Cref{lem:UD_S_2TPD}. The arguments for paths hold just the same for trails.
\begin{lemma}[Strict Undirected Trails] \label{lem:UD_S_2TTD}
    For $k=2$, we can solve \TTEECnoK{S-} for undirected temporal graphs in $\mathcal{O}(\lvert\ecal\rvert + \lvert V\rvert)$ time.
\end{lemma}

Now, we proceed with the computational hardness for exact egde-cover of size two with non-strict trails in directed and undirected graphs.
    \fi

    \iflong
    \subsubsection{Two Non-Strict Trails}
    \fi
    To show hardness for non-strict trails, we adjust the proof of \citet[Theorem 10]{MarinoEulerian2023}.
    They show that \textsc{Eulerian Trail} -- finding a single temporal trail visiting all \textbf{static} edges of a temporal graph \iflong with lifetime two\fi -- is \NP-hard. 
    \iflong 
    \begin{theorem}
    \fi
    \ifshort
    \begin{theorem}[$\star$]
    \fi
    \label{thm:UD_NS_2TTD}
        \TTEECwithK{2}{NS-} is \NP-complete on both directed and undirected temporal graphs.
    \end{theorem}
    \iflong
        \begin{proof}
    As mentioned above, this proof is nearly identical to the proof of \cite[Theorem 10]{MarinoEulerian2023}. We focus on the directed case. For the undirected case, the reduction is exactly the same just that the edges are not directed. The argument is analogous and we will point out the differences. We reduce from \naesat{3}. 

    \vspace{0.4em}\noindent\textit{Construction.}\quad
    Let $\varphi = \bigwedge_{i\in[\ell]}(\bigvee_{j\in[3]}L_{i,j})$ for some $\ell\in\mathbb{N}$ be a \naesat{3} instance with $h=\lvert var(\varphi)\rvert$ many variables and $\ell$ many clauses. We construct a directed temporal graph $\gcal := \tuple{G, \ecal}$ with lifetime three.
    
    The graph consists of an \emph{assignment-gadget}, a \emph{clause-gadget} for each clause $C_i$ in $\varphi$, and two \emph{trail-restriction-gadgets}.
    The assignment-gadget and clause-gadgets contain only edges at time step three.
        The first trail-restriction-gadget contains edges at time step one and makes sure that the first trail has to start at a specific vertex at time step three and has to visit all static edges that need to be covered by the second trail in clause-gadgets at time step three.
        The second trail-restriction-gadget is analogous at time step two for the second trail.

    Let us now describe the three gadgets separately, starting with the assignment-gadget. 
        \begin{figure}[h]
        \centering
        \includegraphics[width=\columnwidth]{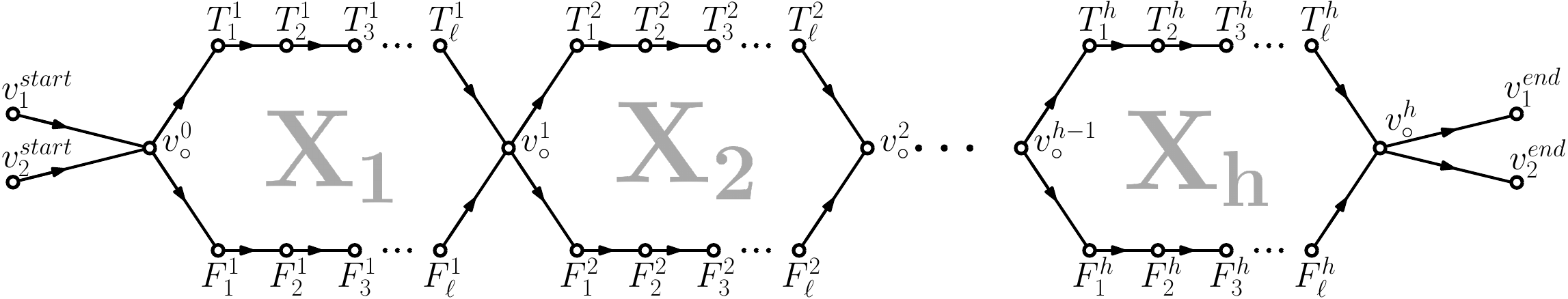}
        \caption{Illustration of the assignment gadget for the \NP-hardness reduction in \Cref{thm:UD_NS_2TTD}.}
        \label{fig:SD_TTD_var_gadget}
    \end{figure}
    Refer to \Cref{fig:SD_TTD_var_gadget} for an illustration.
    The vertices of the assignment-gadget are: 
    \begin{itemize}
        \item vertices $v_1^{start}, v_2^{start}$ and $v_1^{end}, v_2^{end}$, where the two trails are intended to start and end their journey at time step three; 
        \item vertices $v_\circ^{0}, v_\circ^{1}, \ldots, v_\circ^{h}$ that connect separate variable gadgets;
        \item vertices $T^i_j$ and $F^i_j$ for each $i\in [h]$ and $j\in [\ell]$ representing assigning ;\footnote{Technically we only need $T^i_j$ if the clause $C_j$ contains the literal $X_i$ and $F^i_j$ if it contains the literal $\overline{X_i}$, but we keep all of them for the sake of exposition.}
    \end{itemize}
    The edges of the assignment-gadget are (recall that for the undirected case, we have exactly the same edges but without a direction): 
    \begin{itemize}
        \item $((v_1^{start}, v_\circ^{0}),3), ((v_2^{start}, v_\circ^{0}),3)$ and $(( v_\circ^{h}, v_1^{end}),3), (( v_\circ^{h}, v_2^{end}),3)$;
        \item for all $i\in [h]$, the edges $((v_\circ^{i-1}, T^i_1),3), ((v_\circ^{i-1}, F^i_1),3)$ and $((T^i_{\ell},v_\circ^{i}, ),3), ((F^i_{\ell},v_\circ^{i}),3)$;
        \item for all $i\in [h]$ and $j\in [\ell-1]$, the edges $((T^i_j, T^i_{j+1}), 3)$ and $((F^i_j, F^i_{j+1}), 3)$.
    \end{itemize}

        \begin{figure}[t]
        \centering
        \includegraphics[width=\columnwidth]{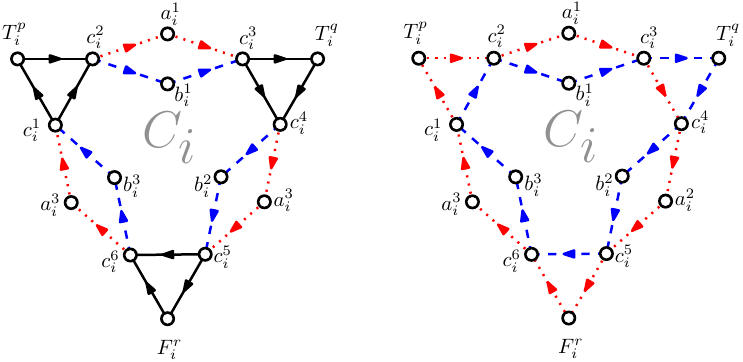}
        \caption{On the left is an illustration of a clause gadget for a clause $C_i = X_p\vee X_q \vee \overline{X_r}$ for the \NP-hardness reduction in \Cref{thm:UD_NS_2TTD}. The red dotted edges have labels $2$ and $3$ and are traversed in time step $2$ by the second trail. The blue dashed edges have labels $1$ and $3$ and are traversed in time step $2$ by the second trail. The black edges have only label $3$. On the right is an example how one can traverse the clause-gadget for $C_i$ in the time step three if one of $T_i^p$ or $F_i^r$ is on the first (red dotted) trail and $T_i^q$ is on the second (blue dashed) trail in the assignment gadget.}
        \label{fig:SD_TTD_clause_gadget}
    \end{figure}
    Now for each $i\in [\ell]$, we have a clause-gadget. Let $C_i$ be a clause over variables $X_p, X_q, X_r$, let $L_i^j$ for $j\in \{p,q,r\}$ be the vertex $T_i^j$ if $X_j$ appears in $C_i$ positively and $F_i^j$ if it appears negatively.
    In addition assume that $p < q < r$. 
    Then the clause-gadget for $C_i$ is as follows (see Figure~\ref{fig:SD_TTD_clause_gadget} for an example when $X_p$ and $X_q$ appear in $C_i$ positively and $\overline{X_r}$ negatively). We have additional vertices $c_i^1, \ldots, c_i^6$, $a_i^1, a_i^2, a_i^3$, and $b_i^1, b_i^2, b_i^3$, these are fresh new vertices not appearing in assignment-gadget or other clause-gadgets. The edges of the clause-gadget are 
    \begin{itemize}
        \item $((c_i^1, L_i^p),3),((L_i^p, c_i^2),3), ((c_i^3, L_i^q),3),((L_i^q, c_i^4),3)$, $((c_i^5, L_i^r),3),((L_i^r, c_i^6),3)$; 
        \item for $j\in \{1,3,5\}$, the edge $((c_i^j, c_i^{j+1}), 3)$;
        \item for $j\in \{1,2,3\}$, the edges $((c_i^{2j}, a_i^j), 3), ((c_i^{2j}, b_i^j), 3)$ and $((a_i^j, c_i^{2j+1}), 3), (( b_i^j, c_i^{2j+1}), 3)$, where $c_i^7 = c_i^1$.
    \end{itemize}
        An important idea to note here is that the trail-restriction-gadgets, will force that all edges $((c_i^{2j}, a_i^j), 3)$ and $((a_i^j, c_i^{2j+1}), 3)$ in all clause gadgets are taken by the same trail, which we refer to as the ``first'' trail and all edges $((c_i^{2j}, b_i^j), 3)$ and $((b_i^j, c_i^{2j+1}), 3)$ in all clause gadgets are taken by the ``second'' trail. Hence both trails have to enter the clause gadget. 
        For the first trail-restriction-gadget, we add additional vertices $s_1$ and for each $j\in \{1,2\}$ and $i\in [\ell]$, vertex $x_i^j$. The edges in this gadget are 
        \begin{itemize}
            \item $((s_1, c_1^2), 1)$, and $((c_\ell^1, v_1^{start}), 1)$;
            \item for all $j\in \{1,2,3\}$ and $i\in [\ell]$, the edges $((c_i^{2j}, b_i^j), 1)$ and $((b_i^j, c_i^{2j+1}), 1)$, where $c^7_i = c^1_i$;
            \item for all $j\in \{1,2\}$ and $i\in [\ell]$, the edges $((c_i^{2j+1}, x_i^j), 1)$;
            \item for all $i\in [\ell-1]$, the edge $((c_i^1, c_{i+1}^2), 1)$.
        \end{itemize}  
        The second trail-restriction-gadget is analogous, but all edges are at time step two and instead of $((c_i^{2j}, b_i^j), 1)$ and $((b_i^j, c_i^{2j+1}), 1)$ there are edges $((c_i^{2j}, a_i^j), 2)$ and $((a_i^j, c_i^{2j+1}), 2)$. 
    \begin{figure}[t]
        \centering
        \includegraphics[width=\columnwidth]{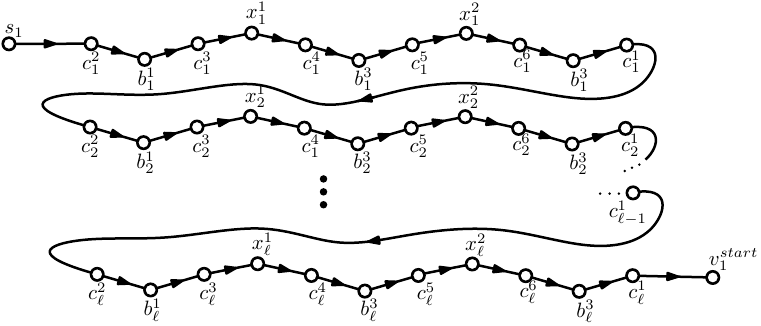}
        \caption{Illustration of the first trail-restriction-gadget for the \NP-hardness reduction in \Cref{thm:UD_NS_2TTD}. The second trail-restriction-gadget is analogous, but edges have time step two, it starts at $s_2$, ends at $v_2^{start}$ and it goes through $a_i^j$ whenever the first trail-restriction-gadget goes through $b_i^j$.}
        \label{fig:SD_TTD_restriction_gadget}
    \end{figure}
    
    \bigparagraph{($\Rightarrow$)\quad} 
    Let $\beta$ be a satisfying assignment for $\varphi$. For a clause $C_i$, let $L_{i,T}$ be the literal in $C_i$ that is satisfied and $L_{i,F}$ be literal that is not satisfied. We construct the two trails as follows. 
            
    The first trail starts in $s_1$ at time step one it follows the first trail-restriction-gadget (which is a path) to $v_1^{start}$. Then at time step three, it goes from $v_1^{start}$ to $v_\circ^0$, for every $i\in [h]$, if $\beta(X_i) = 1$, then the first path follows $P^T_i = (v_\circ^{i-1}, T_1^i, T_2^i, \ldots, T_\ell^i, v_\circ^{i-1})$ and if $\beta(X_i) = 0$ the first path follows $P^F_i = (v_\circ^{i-1}, F_1^i, F_2^i, \ldots, F_\ell^i, v_\circ^{i-1})$. Whenever the first path visits a vertex $L_i^j\in \{T^j_i, F^j_i\}$ that represents the literal $L_{i,T}$ (vertex $T^j_i$, if $L_{i,T} = X_j$  and $F^j_i$, if $L_{i,T} = \overline{X_j}$), then the first path makes a detour that starts and end in $L_i^j$ and is as follows. Let us assume that $X_j$ is the minimum variable in $C_i$, so $L_i^j$ is adjacent to $c_i^1$ and $c_i^2$, the other cases are symmetric. The detour of first trail does the following cycle inside the clause-gadget for $C_i$ (see the red dotted cycle on right in Figure~\ref{fig:SD_TTD_clause_gadget}): It starts with $L_i^j, c_i^2, a_i^1, c_i^3$, afterwards if the vertex $L^q_i$ adjacent to $c_i^3$ and $c_i^4$ represents the literal $L_{i,F}$, then the first trail goes directly to $c_i^4$, else it takes the edge to $L^q_i$ and only then the edge to $c_i^4$, then it continues to $a_i^2$ and $c_i^5$. 
    From $c_i^5$ it again either reaches $c_i^6$ directly if vertex $L^r_i$ adjacent to $c_i^5$ and $c_i^6$ represents the literal $L_{i,F}$ or through the vertex $L^r_i$. The detour finishes by taking $c_i^6, a_i^3, c_i^1, L_i^j$. 
    Note that the remaining edges in the clause-gadget for $C_i$ that are not covered by the first trail form a (directed) cycle that contains a vertex representing the literal $L_{i,F}$ and none of the edge on this cycle is active at time step two (i.e., they can be covered by the second trail). The first trail finishes by taking the edge $((v_\circ^h,v_1^{end}),3)$. 

    The second trail is analogous to the first one. 
    It starts in $s_2$, covers all the edges on the path from $s_2$ to $v_2^{start}$ that are active in time step two. 
    At time step three, it goes from $v_2^{start}$ to $v_\circ^0$, for every $i\in [h]$, if $\beta(X_i) = 0$, then the second trail follows $P^T_i$ ($P^F_i$ is covered by the first trail) and if $\beta(X_i) = 1$ the second path follows $P^F_i$ ($P^T_i$ is covered by the first trail). 
    Whenever the second trail visits vertex $L_i^j$ representing the literal $L_{i,F}$, then it makes a detour to cover the (directed) cycle left in the clause-gadget for $C_i$ by the first trail that contains $L_i^j$ (see the blue dashed cycle on the right in Figure~\ref{fig:SD_TTD_clause_gadget}). The second trail finishes by taking the edge $((v_\circ^h,v_2^{end}),3)$. 

    Given the two trails, it is rather straightforward to verify that each of them is indeed a temporal trail and that they cover the edges of $\gcal$ exactly. 
    
    \bigparagraph{($\Leftarrow$)\quad}
        We start with some general observations about the two trails and then show how to obtain a satisfying assignment for $\varphi$. Observe that at time step three, there are four vertices of degree one -- $v_1^{start}$, $v_2^{start}$, $v_1^{end}$, $v_2^{end}$. Therefore, after time step two each the two temporal trails have to be two different vertices among these four. 
        In addition, $v_1^{end}$, $v_2^{end}$ have degree zero at time steps one and two. 
        We now argue that one trail has to start at time step one in $s_1$, cover all the edge active at time step one and end in $v_1^{start}$ at the beginning of time step two 
        and the second start at time step two in $s_2$, covers all the edges active at time step two and end in $v_2^{start}$ at the beginning of the time step three.
        The graph induced by time step one is a path between $s_1$ and $v_1^{start}$. If there is only one trail active at time step one, then it has be go between $s_1$ and $v_1^{start}$, it the directed case, it has to follow the direction, start in $s_1$ and end in $v_1^{start}$. 
        In the undirected case, it could start in $v_1^{start}$ and finish in $s_1$, but $s_1$ has degree zero at time steps two and three, hence it contradicts above claim that each of the two paths is at one of the vertices $v_1^{start}$, $v_2^{start}$, $v_1^{end}$, $v_2^{end}$ at the start of time step three. 
        The other possibility to cover the path at time step one is for one trail to start at $s_1$ and reach some vertex $v$ and the other to start at $v$ are reach $v_1^{start}$ or start at $v_1^{start}$ and reach $v$ (as we already argued neither of the two trails can finish time step one at $s_1$). 
        In either case, to cover all edges at time step two, and because $s_2$ has degree one at time step two, one of the two trails would have to finish the time step two at $s_2$, which is a contradiction as $s_2\notin \{v_1^{start},v_2^{start},v_1^{end},v_2^{end}\}$. Note that the directed case is even simpler, as a path has to start at $s_2$ to cover the temporal edge $((s_2, c_1^2),2)$. It follows that in order to exactly cover $\gcal$ by at most two trails.
    \begin{itemize}
        \item The first trail starts in $s_1$ covers all edges active at time step one and reaches $v_1^{start}$ at the end of time step one, where it waits at time step two and from where it continues at time step three. At time step three, the first trail finishes at either $v_1^{end}$ or $v_2^{end}$.
        \item The second trail starts in $s_2$ covers all edges active at time step two and reaches $v_2^{start}$ at the end of time step two, from where it continues at time step three. At time step three, the first trail finishes at either $v_2^{end}$ or $v_1^{end}$.
    \end{itemize}
    Given this, the rest of the proof follows by exactly the same arguments as \cite[Theorem 10]{MarinoEulerian2023}, which we include in order for the proof to be self-contained.

    Let us now argue that in both directed and undirected cases, for all $i\in [h]$ the path $P^T_i = (v_\circ^{i-1}, T_1^i, T_2^i, \ldots, T_\ell^i, v_\circ^{i-1})$ either fully belongs to the first trail or it fully belongs to the second trail. 
    Let us consider an internal vertex $T_j^i$ of this path. 
    Since $T_j^i\notin \{v_1^{start},v_2^{start},v_1^{end},v_2^{end}\}$, it is rather straightforward that both edges on $P^T_i$ incident on $T_j^i$ at time step three belong to the same trail if they are indeed the only two edges incident on $T_j^i$ at time step three. 
    Else, the clause $C_j$ contains $X_i$ as a literal and the degree of $T_j^i$ is four with two neighbours $c_j^{2x-1}, c_j^{2x}$ for some $x\in \{1,2,3\}$. 
    Both  $c_j^{2x-1}, c_j^{2x}$ have degree four - one edge between $c_j^{2x-1}$ and $c_j^{2x}$, one edge between $c_j^{2x-1}$ (resp., $c_j^{2x}$) and $T_j^i$, one edge between $c_j^{2x}$ and $a_j^x$ (resp., $c_j^{2x-1}$ and $a_j^{x-1}$) and one edge between $c_j^{2x}$ and $b_j^x$ (resp., $c_j^{2x-1}$ and $b_j^{x-1}$). 
    Note that edge between $c_j^{2x}$ and $a_j^x$ has to be covered by the first trail and the edge between between $c_j^{2x}$ and $b_j^x$ by the second trail. Same for the edges between $c_j^{2x-1}$ and vertices $a_j^{x-1}$ and $b_j^{x-1}$, respectively.
    Hence, the edge between  $c_j^{2x-1}$ and $c_j^{2x}$ has to be covered by a different trail that the edges between $T^i_j$ and the vertices $c_j^{2x-1}$ and $c_j^{2x}$, respectively. 
    It follows that the edges between $T^i_j$ and the vertices $c_j^{2x-1}$ and $c_j^{2x}$, respectively, have to be covered by the same trail. 
    As $T^i_j$ can only be inner vertex of either of the two trails, also both the edges in $P^T_i$ incident on $T^i_j$ belong to the same trail. Analogous argument shows that for all $i\in [h]$ the path $P^F_i = (v_\circ^{i-1}, F_1^i, F_2^i, \ldots, F_\ell^i, v_\circ^{i-1})$ either fully belongs to the first trail or it fully belongs to the second trail. 
    
    Now let us show that $P^T_i$ belongs to the first trail if and only if $P^F_i$ belongs to the second trail and, vice-versa, $P^F_i$ belongs to the first trail if and only if $P^T_i$ belongs to the second trail.  
    For all $i \in \{0,1,\ldots, h\}$, the vertex $v_\circ^i$ has degree four and can only be an inner vertex of either trail at time step three. Clearly, $((v_1^{start}, v_\circ^0),3)$ belongs to the first trail and  $((v_2^{start}, v_\circ^0),3)$ belongs to the second trail. It follows that edges $((v_\circ^0, T^1_1), 3)$ and $((v_\circ^0, F^1_1), 3)$ belong to different trails and so $P^T_1$ and $P^F_1$ belong to different trails as well. Now by an inductive argument, if $P^T_i$ and $P^F_i$ belong to different trails (i.e., $((T^i_\ell, v_\circ^i, ), 3)$ and $((F^i_\ell,v_\circ^i), 3)$ are in different trails), then $((v_\circ^i, T^{i+1}_1), 3)$ and $((v_\circ^i, F^{i+1}_1), 3)$ belong to different trails and so $P^T_{i+1}$ and $P^F_{i+1}$ belong to different trails as well.

    Let us now define an assignment $\beta$ such that $\beta(X_i)= 1$ if and only if $P_i^T$ belongs to (i.e., is fully contained in) the first trail. Note that if $\beta(X_i)= 1$, then $P_i^F$ belongs to the second trail. Similarly, if $\beta(X_i)= 0$, then $P_i^F$ belongs to the first trail and $P_i^T$ belongs to the second trail. Let $C_j = L_j^p\vee L_j^q\vee L_j^r$ be a clause of $\varphi$. Because the edge $(c_j^2,a_j^1)$ is at time step two in second trail, it has to be at time step three in the first trail. It follows that the first trail has to enter the clause-gadget for $C_j$ at time step three. The only three entry point to the clause-gadget are the vertices representing $L_j^p$, $L_j^q$, and $L_j^r$. The two edge incident with $L_j^p$ (resp., $L_j^q$ or $L_j^r$) that are not in the clause-gadget for $C_j$ are either on $P^T_p$ (resp., $P^T_q$ or $P^T_r$) or $P^F_p$ (resp., $P^F_q$ or $P^F_r$) depending on whether the $L_j^p$ (resp., $L_j^q$ or $L_j^r$) is positive or negative. Since the first trail contains $P^T_i$ if and only if $\beta(X_i) = 1$ and it contains $P^F_i$ if and only if $\beta(X_i) = 0$, it follows that since the first trail enters the clause-gadget for $C_j$, at least one of the literals $L_j^p$, $L_j^q$, and $L_j^r$ has to be true under $\beta$. Analogously, the second trail has to enter the clause gadget for $C_j$, because of the edge $(c_j^2,b_j^1)$ that cannot be covered by the first trail at time step three, as it is covered by the first trail at time step one. It follows that since the second trail enters the clause-gadget for $C_j$, least one of the literals $L_j^p$, $L_j^q$, and $L_j^r$ has to be false under $\beta$. Using the same argument for each clause $C_j$, $j\in [\ell]$, we get that under the assignment $\beta$, in $C_j$ at least one literal evaluates to True and at least one literal evaluate to False, hence $\beta$ is a satisfying assingment for \naesat{3} instance $\varphi$. 
    \end{proof}
    \fi
    
    Having completely resolved our problem for one and two journeys, we now move on to the computational complexity of the general problem with $k$ at least 3. 
    As we will see, for larger $k$ there is a divergence in the complexity of the problem for the three journey types, which, interestingly, is notably different from the behavior observed for $k=2$.

\section{Paths and Trails -- Computational and Approximation Hardness for $k\geq3$}\label{sec:pathsAndtrails}
We prove that all versions of our problem are hard for both paths and trails, starting with intractability for paths, then augmenting the construction to get a strong inapproximability bound, and finally extending the results to trails. 
    \iflong
    \paragraph{Computational Hardness For Paths}
        We prove the 
        four claims by providing the detailed reduction for an exact edge-cover with \textit{strict} paths in \textit{directed} graphs in \Cref{thm:SD_TPD} and explaining the adjustments for the other problems variants 
        in the following  \Cref{lem:D_NS_TPD,lem:UD_SNS_TPD}.
        \begin{theorem}
        \label{thm:TPD_para}\label{thm:SD_TPD}
            \TPEECnoK{S-} is \NP-hard on directed temporal graphs.
        \end{theorem}
        \begin{proof}
We reduce from \naesat{k}.
    Let $\varphi$ 
    be an instance of  \naesat{k} with $h$ variables $X_1,\dots,X_h$ and $\ell$ clauses. We construct a directed temporal graph $\gcal := \tuple{G, \ecal}$.
    
    Intuitively, \gcal consists of an \textit{assignment-gadget} constructed out of two paths and a \textit{satisfaction-gadget} constructed out of $k$ paths. If $\varphi$ is satisfiable and the two paths in the assignment-gadget are chosen accordingly, they will be extendable in the satisfaction-gadget. If $\varphi$ is not satisfiable, any path eec will need one additional path to cover all edges exactly once and \gcal has a minimal path eec of size $k+1$.
    For an illustration refer to \Cref{fig:paraNPhard}.

    %
    For each variable $X_i$, we construct a \textit{variable-gadget} with vertices $\{v^i,v^{i+1}\}\cup\{T_j^i,F_j^i \colon j\in[\ell]\}$ and static edges 
    $E_T^{x_i} = 
    \left\{ \{(v_i,T^i_1)\} \cup \{ (T^i_j,T^i_{j+1}) \colon j\in[\ell] \right\} \cup \{(T^i_\ell,v_{i+1})\}$ and
    $
    E_F^{x_i} = 
    \left\{ \{(v_i,F^i_1)\} \cup \{ (F^i_j,F^i_{j+1}) \colon j\in[\ell] \right\} \cup \{(F^i_\ell,v_{i+1})\}$. 
    All variable-gadgets together form the \textit{assignment-gadget}.

    Let $C_i$ be a clause containing $k$ literals.
    For each literal $L_\alpha$ in $C_i$, let $X_\alpha$ be the corresponding variable.
    We construct a \textit{clause-gadget} with vertices $\{c_{i},c_{i+1}\}$ and static edges $\{ (c_i,F^\alpha_i), (F^\alpha_i,c_{i+1})\}$ for every positive literal $L_\alpha\equiv X_\alpha$ and edges
    $\{ (c_i,T^\alpha_i),(T^\alpha_i,c_{i+1})\}$ for every negative literal $L_\alpha\equiv\overline{X_\alpha}$.
    All clause-gadgets together form the \textit{satisfaction-gadget}.

    \begin{figure*}[t]
        \centering
        \includegraphics[width=0.7\textwidth]{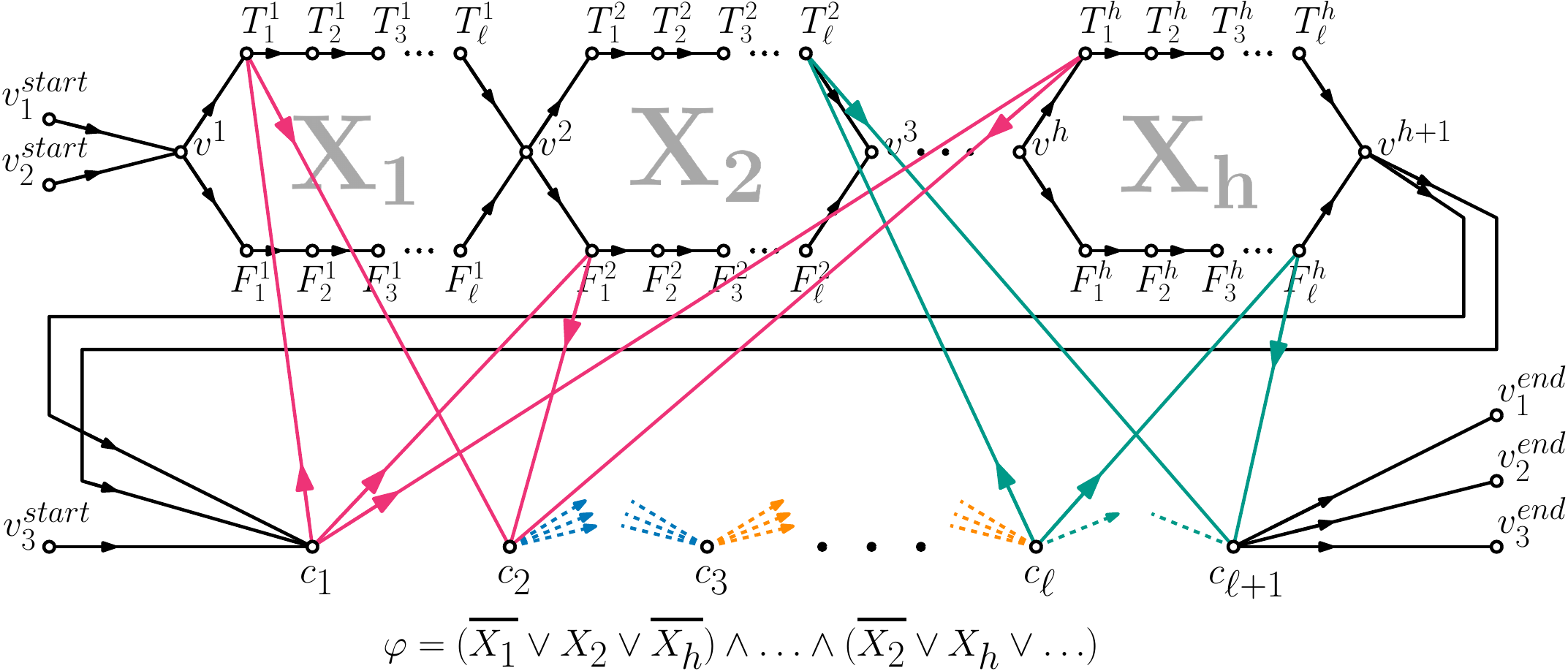}
        \caption{Illustration of the construction for the \NP-hardness reduction in \Cref{thm:SD_TPD} for $k=3$. The assignment-gadget is on the top and the satisfaction-gadget is on the bottom. The two edges $(v^{h+1},c_1,2h(\ell+1)+1),(v^{h+1},c_1,2h(\ell+1)+2)$ connecting the assignment- and satisfaction-gadget are drawn as two separate edges. With this each edge has exactly one time label increasing with the direction of the edges.
        }
        \label{fig:paraNPhard}
    \end{figure*}
    Lastly, we add $k$ many \textit{start-vertices} and \textit{end-vertices} $v^{start}_i,v^{end}_i, i\in[k]$.
    We connect the start-vertices $v_1^{start}$ and $v_2^{start}$ with $v^1$ and the remaining $k-2$ start-vertices with the first clause-vertex $c_1$.
    We connect the assignment-gadget and the satisfaction-gadget by adding two edges $(T^h_\ell,c_1)$ and $(F^h_\ell,c_1)$.
    And lastly, we connect the last clause-vertex $c_{l+1}$ with the $k$ end-vertices $v_1^{end},\dots v_k^{end}$.
    Again, refer to the illustration in \Cref{fig:paraNPhard}. 

    To each edge, we incrementally add distinct time labels going through the variable-gadgets first and the clause-vertices next -- following the direction of the edges. For readability, we give the edges going out of the start-vertices all time-label zero. We now present all time labels:
    
    Let $D_i=i\cdot 2(\ell+1)$ be the largest time-label added for variable $X_i$, \ie $D_0=0$. For each variable $X_i$ add two incremental temporal paths, namely, the True-path $\tdirect{v^{i}}{T^i_1}{D_{i-1}+1}$, $\tdirect{T^i_j}{T^i_{j+1}}{D_{i-1}+j+1}$ for $1\leq j<\ell$ and $\tdirect{T^i_\ell}{v^{i+1}}{D_{i-1}+\ell+1}$, and the False-path $\tdirect{v^{i}}{F^i_1}{D_{i-1}+\ell+2}$, $\tdirect{F^i_j}{F^i_{j+1}}{D_{i-1}+\ell+j+1}$ for $1\leq j<\ell$ and $\tdirect{F^i_\ell}{v^{i+1}}{D_{i-1}+2(\ell+1)}$.

    Next, add labels to the connections $\tdirect{v^{h+1}}{c^1}{D_h+1}$, $\tdirect{v^{h+1}}{c^1}{D_h+2}$ and let $E=D_h+3$ be the first time label of the satisfaction-gadget.
    For every clause $C_i=L_{i,1}\vee\dots\vee L_{i,k}$ with $1\leq i\leq \ell$ and $1\leq j\leq k$ we add the following labels.
    If $L_{i,j}$ is positive we add $\tdirect{c_i}{F^j_i}{E+(i-1)2k+2(j-1)}$ and $\tdirect{F^j_i}{c_{i+1}}{E+(i-1)2k+2j}$.
    If the literal $L_{i,j}$ is negative, we add $\tdirect{c_i}{T^j_i}{E+(i-1)2k+2(j-1)}$ and $\tdirect{T^j_i}{c_{i+1}}{E+(i-1)2k+2j}$.
    This forms $k$ many paths from $c_i$ to $c_{i+1}$.
    Lastly, we add label $E+(\ell+1)2k+j$ to the edge from $c_{\ell+1}$ to end-vertex $v^{end}_j$.

    This concludes the construction, which can be done in polynomial time.    
    Note that every static edge has exactly one distinct time label. Therefore, there is a distinct order for the edges and we omit time-labels during the analysis. 
    
    \bigparagraph{($\Rightarrow$)\quad} 
    Let $\beta$ be a satisfying assignment for $\varphi$. We construct $k$ paths, starting with $p_1$ and $p_2$ that cover 
    $(v_1^{start},v^1)$ and $(v_2^{start},v^1)$, respectively.

    For each variable $X_i$, $i\in[h]$, we extend $p_1$ and $p_2$ according to $\beta$.
    If $\beta(X_i)=1$, we extend
    $p_1$ by $(v^{i},T^i_1)$, $\{(T^i_j,T^i_{j+1})\colon j\in[\ell]\}$, $(T^i_\ell,v^{i+1})$ and
    $p_2$ by $(v^{i},F^i_1)$, $\{(F^1_j,F^1_{j+1})\colon j\in[\ell]\}$, $(F^i_\ell,v^{i+1})$.
    Invertedly, if $\beta(X_i)=0$, we extend $p_1$ over the False-vertices and $p_2$ over the True-vertices.
    After processing the last variable $X_h$, both paths are extended by the edge to $c_1$ and we construct the other $k-2$ paths $p_3,\dots,p_k$ starting with $(v_3^{start},c_1),\dots,(v_k^{start},c_1)$, respectively.
    
    Since $\beta$ is a satisfying assignment for the \naesat{k} formula, for each clause $C_i$ there is one satisfied literal and one unsatisfied literal. We assume without loss of generality that these are the first and the second literal $L_{i,1}$ and $L_{i,2}$.
    If $L_{i,1}=X_j$ is a positive literal, we extend $p_1$ by $(c_i,F^j_1)$ and $(F^j_1,c_{i+1})$. Since $p_1$ visited the True-vertices of $X_j$, it has not visited $F^j_1$ before and this does not violate the path condition. Conversely, if $L_{i,1}=\overline{X_j}$ is a negative literal, we extend $p_1$ by $(c_i,T^j_1)$ and $(T^j_1,c_{i+1})$. We proceed in the same fashion for $p_2$ and $L_{i,2}$.    
    The remaining $k-2$ outgoing edges of the vertex $c_i$ can be covered by $p_3$ to $p_{k}$ since these paths did not visit any vertices outside of the satisfaction-gadget.
    We end every path $p_i$ with $(c_{\ell+1},v^{end}_i)$.
    
    With this every edge is covered exactly once and $\pcal=\{p_1,\dots,p_k\}$ is a path exact edge-cover of size $k$.
    
    \bigparagraph{($\Leftarrow$)\quad} 
    Let $\mathcal{P}=\{p_1,\dots,p_k\}$ be a path eec of size $k$ for \gcal. We define an assignment $\beta$ for $\varphi$ by $\beta(X_i)=1$ if and only if $p_1$ visits $T^i_1$ via $(v^{i},T^i_1)$.
    
    Observe that the assignment-gadget with the start-vertices $v^{start}_1$ and $v^{start}_2$ form a temporal subgraph of \gcal with a path eec of size $2$. It consists of two edge-disjoint paths $p_1,p_2$ that each contain one of the two False-/True-paths connecting consecutive variable-gadgets.
    Furthermore, there have to be at least $k-2$ additional temporal paths starting in $v_i^{start}$, $3\leq i \leq k$, as these vertices have out-degree one and no incoming edges.
    \iflong
    \begin{figure*}[h]
        \centering
        \includegraphics[width=0.7\textwidth]{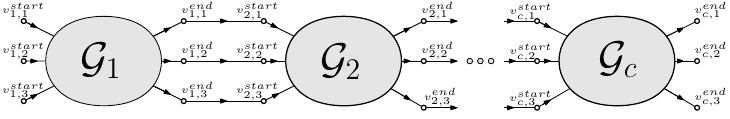}
        \caption{Illustration of the construction used to prove \Cref{thm:approx_TPD} with $c$ copies $\gcal_1,\dots,\gcal_c$ of $\gcal$.
        The last clause vertex $c_{i,l+1}$ of $\gcal_i$ is connected twice to the first variable gadget via $v_{i+1,1}$ and $k-2$ times to $c_{i+1,1}$ from $\gcal_{i+1}$.
        }\label{fig:approx_TPD_proof}
    \end{figure*}
    \fi    
    In the satisfaction gadget, $k$ paths meet in the clause-vertex $c_1$ and each has to cover one outgoing edge of $c_1$. 
    Since $\mathcal{P}$ is an eec of size $k$, $p_1$ and $p_2$ have to be paths. So, for each clause-vertex $c_i$, there is one outgoing edge corresponding to a literal $L_{i,j_1}$ that is covered by $p_1$ and another corresponding to a literal $L_{i,j_2}$ that can be covered by $p_2$.
    The former implies $L_{i,j_1}$ is satisfied while the latter implies $L_{i,j_2}$ is not satisfied:
    For the former, $p_1$ must have visited the True-path of the variable-gadget of $X_1$ if $L_{i,j_1}=X_1$ and the False-path of the variable-gadget of $X_1$ if $L_{i,j_1}=\overline{X_1}$. Therefore, $\beta(X_1)=1$ if $L_{i,j_1}=X_1$ and $\beta(X_1)=0$ if $L_{i,j}=\overline{X_1}$.
    For the latter, $p_2$ must have visited the True-path of the variable-gadget of $X_2$ if $L_{i,j_2}=X_2$ and $p_1$ visited the False-path. Analogously, if $L_{i,j_2}=\overline{X_2}$. Therefore, $\beta(X_2)=0$ if $L_{i,j_2}=X_2$ and $\beta(X_2)=1$ if $L_{i,j_2}=\overline{X_2}$. 
    For each clause there is one satisfied and one unsatisfied literal, so the assignment with $\beta(X_i)=1$ if and only if $p_1$ visits $T^i_1$, satisfies the \naesat{k} formula.
    \end{proof}
    
    Since each time label in the construction appears exactly once, every non-strict temporal path is also a strict path, which directly implies \NP-hardness for \TPEECnoK{NS-} on directed graphs.
    \begin{corollary} \label{lem:D_NS_TPD}
        \TPEECnoK{NS-} is \NP-hard on directed temporal graphs.
    \end{corollary}
    For undirected graphs, the \NP-hardness follows from the distinct temporal order in which the edges appear. Observe that all edges in the assignment-gadget appear  before all edges of the clause-gadget, which enforces a direction on any eec of minimal size, regardless of whether the paths are strict or non-strict. So, replacing every directed edge by an undirected edge yields a construction on undirected graphs.
    \begin{corollary} \label{lem:UD_SNS_TPD}
        \TPEECnoK{(N)S-} are \NP-hard on undirected temporal graphs.
    \end{corollary}

    \fi
    \ifshort
    \begin{theorem}[$\star$]
    \label{thm:TPD_para}\label{thm:SD_TPD}
        \TPEECnoK{(N)S-} is \NP-complete on directed and undirected temporal graphs, for every $k\geq 3$. 
    \end{theorem}
    
\begin{proof}[Proof sketch]
    All four constructions\,---\,(non-)strict paths on (un)directed graphs\,---\,follow the same idea, inspired by the \textsc{SAT} reduction in \citet[Theorem 2]{klobas_interference-free_2023}.
    We outline the case of non-strict paths in directed graphs, which is the most concise construction, requiring just two time steps.
    
    We reduce from \naesat{k}, which is known to be NP-complete \cite{schaefer1978complexity}. The input is a formula $\varphi$ over $\ell$ variables in conjunctive normal form where each clause contains exactly $k$ literals. The goal is to check whether the formula is satisfiable such that each clause has at least one literal that is True and one that is False. 
    
    We construct a temporal graph \gcal with an exact edge-cover of size $k$ if and only if $\varphi$ is satisfiable; otherwise, the cover will require size $k+1$. See \Cref{fig:paraNPhard} for an illustration. 

    For each variable $X_i$, we construct a \textit{variable-gadget} 
    with vertices $\{v^{i},v^{i+1}\} \cup \{T^{i}_j,F^i_j\colon j\in[\ell]\}$ and edges 
    $\ecal_T^{x_i} = 
    \left\{ \{(v^i,T^i_1,1)\} \cup \{ (T^i_j,T^i_{j+1}) \colon j\in[\ell] \right\} \cup \{(T^i_\ell,v^{i+1},1)\}$ and
    $
    \ecal_F^{x_i} = 
    \left\{ \{(v^i,F^i_1,1)\} \cup \{ (F^i_j,F^i_{j+1}) \colon j\in[\ell] \right\} \cup \{(F^i_\ell,v^{i+1},1)\}$. All variable-gadgets together form the \textit{assignment-gadget}.

    Let $C_i$ be a clause containing $k$ literals.
    For each literal $L_\alpha$ in $C_i$, let $X_\alpha$ be the corresponding variable.
    We construct a \textit{clause-gadget} with vertices $\{c_{i},c_{i+1}\}$ and edges $\{
    (c_i,F^\alpha_i,2),(F^\alpha_i,c_{i+1},2)\}$ for every positive literal $L_\alpha\equiv X_\alpha$ and 
    edges $\{
    (c_i,T^\alpha_i,2),(T^\alpha_i,c_{i+1},2)\}$ for every negative literal $L_\alpha\equiv\overline{X_\alpha}$. All clause-gadgets together form the \textit{satisfaction-gadget}.
    Lastly, we connect the assignment-gadget and the satisfaction-gadget by adding $\{(v^{h+1},c_1,1),(v^{h+1},c_1,2)\}$.

    Intuitively, the \textit{assignment-gadget} is constructed out of two paths and the \textit{satisfaction-gadget} is constructed out of $k$ paths.
    If $\varphi$ is satisfiable and the two paths in the assignment-gadget are chosen correctly, then they will be extendable to cover two of the $k$ paths of the satisfaction-gadget. The remaining paths can be covered for free. If $\varphi$ is not satisfiable, we will need to create an additional path no matter how the two paths in the assignment-gadget are chosen.
    \end{proof}

    \fi

    \paragraph{Hardness of Approximation.}
    Further extending this construction, we show that there cannot be a polynomial time algorithm that computes a constant $\alpha$-approximation of \TPEECnoK{(N)S-} for $k\geq 3$, unless $\PP=\NP$.
    
    \ifshort
    We do so by connecting a sufficient number of copies of the constructed \gcal  via edges from the end-vertices of the $j^\text{th}$ copy $v_{j,i}^\text{end}$ to the start-vertices of the $j+1^\text{th}$ copy $v_{j+1,i}^\text{start}$, as illustrated in \Cref{fig:approx_TPD_proof} for \sat{3}.
    Each copy increases the number of paths -- necessary to cover all edges if $\varphi$ is unsatisfiable -- by one. If the number of copies is large enough, we would be able to correctly decide whether the \naesat{k} formula $\varphi$ is satisfiable given an $\alpha$-approximation of \TPEECnoK{(N)S-}.
    \fi
    \iflong 
    \begin{theorem}
    \fi
    \ifshort
    \begin{theorem}[$\star$]
    \fi\label{thm:approx_TPD}
        For all $\alpha > 1$, there is no polynomial time $\alpha$-approximation algorithm for \TPEECnoK{(N)S-} on directed and undirected temporal graphs unless $\PP = \NP$.
    \end{theorem}
    \ifshort
    \begin{figure}[ht]
        \centering
        \includegraphics[width=\columnwidth]{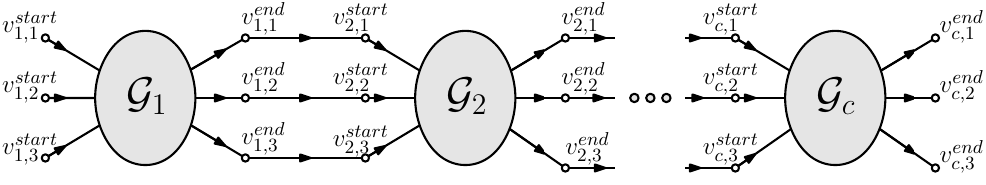}
        \caption{Illustration of \Cref{thm:approx_TPD} with $c$ copies $\gcal_1,\dots,\gcal_c$ of $\gcal$.
        The end-vertices $v_{j,1}^{\text{end}}, v_{j,2}^{\text{end}} ,v_{j,3}^{\text{end}}$ of $\gcal_j$ are connected to start-vertices $v_{j+1,1}^{\text{start}}, v_{j+1,2}^{\text{start}} ,v_{j+1,3}^{\text{start}}$ of $\gcal_{j+1}$.
        }\label{fig:approx_TPD_proof}
    \end{figure}
    \fi
    \iflong
    \begin{proof}
        First, we show that the reduction from \naesat{k} for $k=3$ implies  there is no $\alpha$-approximation algorithm for \TPEECnoK{(N)S-} with $\alpha<\frac{4}{3}$. We then explain how to extend this to arbitrary $\alpha$.
    
        Assume that an $\alpha$-approximation algorithm for \TPEECnoK{(N)S-} with $\alpha<\frac{4}{3}$ exists. Given a \naesat{3} instance $\varphi$, let \gcal be the constructed temporal graph from \Cref{thm:TPD_para} such that a temporal path exact edge-cover of size 3 for \gcal exists if and only if $\varphi$ is satisfiable.

        Thus, if $\varphi$ is satisfiable, the $\alpha$-approximation algorithm returns a path \eecShort of size $k$ with $3\leq k < 4$, \ie between the optimum $3$ and the worst approximation $3\cdot\alpha=3\cdot\frac{4}{3}=4$.
        Whereas, if $\varphi$ is unsatisfiable, the algorithm returns a path \eecShort of size $k'$ with $k'\geq 4$. As the two intervals are disjoint, one could correctly decide \naesat{3} in polynomial time, which would imply $\PP=\NP$. The claim follows.

        For larger $\alpha$, we concatenate multiple copies of $\gcal$, as each copy increases the difference between the optimal path \eecShort for a satisfiable and an unsatisfiable \naesat{3}-instance: Choose a number of copies $c$ such that $c > 3\cdot (\alpha-1)$. Create $c$ copies $\gcal_1, \dots, \gcal_c$ of \gcal. Iteratively connect $\gcal_i$ to $\gcal_{i+1}$ for increasing $0 < i < c$ as follows:
        Connect the end-vertices $v_{i,j}^{end}$ of $\gcal_i$ with the start-vertices $v_{i+1,j}^{start}$ of $\gcal_{i+1}$ via an edge with time label $t_{max}^{i+1}$ (where $t_{max}^i$ is the largest time label used in $\gcal_i$).
        Furthermore, increase all time labels in $\gcal_{i+1}$ by $t_{max}^{i+1}$. See \Cref{fig:approx_TPD_proof} for an illustration.
        \ifshort
        \begin{figure}[ht]
            \centering
            \includegraphics[width=\columnwidth]{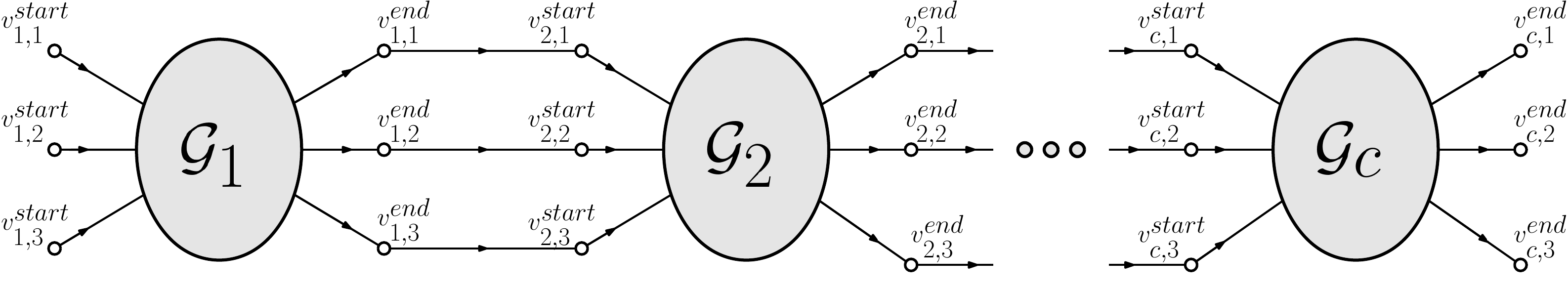}
            \caption{Illustration of the construction used to prove \Cref{thm:approx_TPD} with $c$ copies $\gcal_1,\dots,\gcal_c$ of $\gcal$.
            The last clause vertex $c_{i,l+1}$ of $\gcal_i$ is connected twice to the first variable gadget via $v_{i+1,1}$ and $k-2$ times to $c_{i+1,1}$ from $\gcal_{i+1}$.
            }\label{fig:approx_TPD_proof}
        \end{figure}
        \fi
        If $\varphi$ is satisfiable, each copy $\gcal_i$ can be covered with 3 paths. 
        Because end-vertices and start-vertices are connected, the same 3 paths can be re-used for the entire construction: Each $\gcal_i$ for $i>1$ can extend the walks of $\gcal_{i-1}$ instead of beginning new walks. 
        If $\varphi$ is unsatisfiable, at least $3+c$ paths are needed to cover every temporal edge: If there was an \eecShort of smaller size, there would be a copy $\gcal_i$ within which no new path begins, which in turn corresponds to a satisfying assignment for $\varphi$.
        
        Given a satisfiable $\varphi$, an $\alpha$-approximation algorithm returns a path \eecShort of size $k$ with $3 \leq k \leq 3\cdot\alpha$. Given an unsatisfiable $\varphi$, it returns a path \eecShort of size $k'$ with $3+c \leq k' \leq (3+c)\cdot\alpha$.
        Note that $3\cdot\alpha = 3+3(\alpha-1) < 3+c$, meaning the two intervals are disjoint.
        Hence, \naesat{3} could be correctly decided by checking whether the $\alpha$-approximation returns a value smaller than $3+c$, which would imply $\PP=\NP$.
    
        As in \Cref{thm:TPD_para}, this directly extends to the non-strict variant and to undirected temporal graphs.
    \end{proof}
    \fi

    \paragraph{Extending Paths to Trails.}
    It is straightforward to adjust both constructions such that covering clause-edges forces revisiting edges instead of revisiting vertices in the variable-gadget. 
    \ifshort    
    This implies the following.
    \fi
    \iflong
    More formally, we extend the construction of \Cref{thm:TPD_para} by exchanging every boolean-vertex $B^i_j$, for $b=B\in\{T,F\}$, by a \emph{boolean-edge}, \ie two vertices $b^i_j$, $B^i_j$ connected by an edge $(b^i_j,B^i_j)$ and for any incoming edge $(x,B^i_j)$ we put $(x,b^i_j)$ and for any outgoing edge $(B^i_j,x)$ we put $(B^i_j,x)$.
    The time label of the incoming and outgoing edges stays the same and the time label of $(b^i_j,B^i_j)$ has to be after the incoming edge of $B^i_j$ in the assignment-gadget and before the outgoing edge in the assignment-gadget. All later time labels have to be shifted by one for each added edge.
    \fi 
 
    \begin{corollary}\label{thm:TTD_para}
        \TTEECnoK{(N)S-} are \NP-complete on directed and undirected temporal graphs.
    \end{corollary}
    \begin{corollary}\label{cor:approx_TTD}
        For all $\alpha > 1$, there is no polynomial time $\alpha$-approximation algorithm for \TTEECnoK{(N)S-} on directed and undirected temporal graphs, unless $\PP = \NP$.
    \end{corollary}


\section{Walks -- One Polynomial-Time Algorithm and Three Hardness Reductions for $k\geq3$}
\label{sec:walks} 
    When constructing exact edge-covers using temporal walks, we find a striking contrast to paths and trails: the problem remains polynomial-time solvable for any constant $k$. \iflong This distinguishes it from paths and trails, which become intractable for constant $k\geq3$. \fi
    \iflong 
    \begin{proposition}
    \fi
    \ifshort
    \begin{proposition}[$\star$]
    \fi
    \label{prop:walks_XP}
        \TWEECnoK{(N)S-} can be computed in time $\mathcal{O}(n^{\mathcal{O}(k)})$ which is polynomial for any constant $k$.
    \end{proposition}
    This is due to the nature of temporal walks in exact edge-covers: they are time-respecting, connected journeys that cover each temporal edge exactly once.
    For a constant number $k$, we can brute-force the start-terminals and, at each time step, keep track of all possible positions of the $k$ walks -- there are at most $n^k$ such positions. 

    For an arbitrary size of the eec, the computational complexity varies depending on the variant of the problem. For strict walks on directed graphs, we can optimize the approach from \Cref{prop:walks_XP}, while the problem becomes \NP-complete for non-strict walks on directed graphs and for both strict and non-strict walks on undirected graphs.

    \subsection{Strict Walks in Directed Graphs -- Two Polynomial Time Approaches} \label{subsec:walk_poly}
    For strict walks on directed graphs, \iflong we can deduce a topological order of the vertices imposed by the passage of time and edge direction.
    This implies that\fi the graph transformation introduced by \citet{wu_path_2014} forms a directed acyclic graph. We can adjust it so that a temporal walk eec in \gcal corresponds to an exact path cover (covering every vertex exactly once), which is computable in polynomial time\iflong using a maximum flow algorithm\fi. 
    \iflong
    The transformed graph contains vertices linear in the number of temporal edges and edges quadratic in the number of temporal edges, so it can get quite large.
    Although this approach is not the most efficient, it provides an interesting construction which can be adjusted to solve also other variants of the covering problem.
    \fi

    A faster (linear time) algorithm can be achieved by using the endpoint-tracking approach from \Cref{prop:walks_XP}. Since the walks are strict, two edges at the same time step need to be taken by two different walks, and since the graph is directed, the possible extensions are clearly defined. If an edge cannot extend an existing walk, a new walk starting with that edge has to be introduced. Furthermore, the starting points of the walks at the first time step are uniquely defined.
    This way, at each time step, there is exactly one possible position for walks of an exact edge-cover. 
    \iflong 
    \begin{theorem}
    \fi
    \ifshort
    \begin{theorem}[$\star$]
    \fi\label{thm:SD_TWD_linear_time}
        \TWEECnoK{S-} on directed temporal graphs can be solved in $\mathcal{O}(\lvert\ecal\rvert)$ if the edgestream is given.
    \end{theorem}

    \iflong
    \begin{proof}
    Let $\gcal=\tuple{G, \ecal}$ be a temporal directed graph with lifetime $t_{max}$.
    As we consider strict walks, we assume without loss of generality exactly one edge per time label, $\lvert E_t\rvert = 1$ for all $t\in[t_{max}]$.
    We provide a stream algorithm that computes the walk eec by iterating through the edges in temporal order while keeping track of all possible end-terminals for the walks. For a new edge $((u, v), t)$, any walk ending at $u$ before time step $t$ can be extended. If no walk ends at $u$ before $t$, this edge has to be the start of a new walk.

    Let $\gcal_{\leq t}$ be the graph containing every edge up to time label $t$, and let $\wcal_t$ be a minimal walk eec for $\gcal_{\leq t}$. We show by induction over $t$ that the multiset of end-terminals $\setEndTs{\wcal_t}$ is unique, \ie all minimal eecs have the same end-terminals.
    
    Let \tdirect{u_t}{v_t}{t} be the edge at time $t$. For $t=1$, we have $\setEndTs{\wcal_1}=\{v_1\}$.
    Assume for a fixed $t$, that $\setEndTs{\wcal_t}$ is unique and consider $t+1$.
    If $|\wcal_{t+1}|=|\wcal_t|$, then an existing walk in $\wcal_{t}$ must be extended to cover the new edge.
    No matter which walk ending at $u_{t+1}$ is extended, we have $\setEndTs{\wcal_{t+1}}=\setEndTs{\wcal_t} \setminus \{u_{t+1}\} \cup \{v_{t+1}\}$.
    
    If $|\wcal_{t+1}|=|\wcal_t| + 1$, a new walk had to be added to cover the new edge.
    We show by contradiction that $\setEndTs{\wcal_{t+1}} = \setEndTs{\wcal_t} \cup \{v_t\}$ for any walk eec:
    Assume there was a different eec $\wcal_{t+1}'$ with $\setEndTs{\wcal_{t+1}} \neq \setEndTs{\wcal_{t+1}'}$.
    $v_t$ must still be contained in $\setEndTs{\wcal_{t+1}'}$ as the edge with time label $t+1$ leads to it.
    Removing this edge from $\wcal_{t+1}'$ would yield a new eec $\wcal_t'$ for $\gcal_{\leq t}$ with $\setEndTs{\wcal_t} \neq \setEndTs{\wcal_t'}$, which in turn would violate the induction assumption.

    Knowing that the set of end-terminals is distinct in each step, we can conclude that any time step in which the algorithm begins a new walk is unavoidable: If the algorithm cannot solve a $\gcal_{\leq t}$ without starting a new walk, neither can the optimal solution. Hence, by induction, the partial solution is optimal for each $\gcal_{\leq t}$ and for the entire graph $\gcal$.
\end{proof}
    \fi
    
    

    \subsection{Non-Strict Walks or Undirected Graphs -- Hard to Compute Efficiently} \label{subsec:walk_NP}
    We proceed with the complexity of the other three variants of walk exact edge-covers. 
    For \TWEECnoK{NS-} in directed graphs we prove $\Wtwo$-hardness and in undirected graphs we show that both \TWEECnoK{(N)S-} are \NP-hard.
    
    Towards the first result, observe that non-strict walks can traverse directed cycles appearing at one time step, unlike strict walks.
    Given such a directed cycle in which every edge has the same time label, any non-strict walk covering these edges must either omit some edges of the cycle or return to the vertex is started at.
    We exploit this behaviour of returning to a chosen vertex to cover a cycle,  
    for a reduction from $k$-\textsc{Hitting Set} to \TWEECnoK{NS-} on directed graphs.
    \ifshort
    \fi
    

    \iflong 
    \begin{theorem}
    \fi
    \ifshort
    \begin{theorem}[$\star$]
    \fi \label{thm:DNS_TWD}
        \TWEECnoK{NS-} 
        on directed temporal graphs is \NP-complete and \Wtwo-hard when parameterized by the number of walks in the exact edge-cover.
    \end{theorem}
    \iflong
    \begin{proof}
    We reduce from \hittingset. Given an instance $(N=\{1, \dots n\}, S=\{S_1, \dots S_m\}, k)$, assume that all sets have size 2 or larger, otherwise augment them with auxiliary elements. We construct a temporal graph $\gcal=\tuple{(V,E),\ecal}$ with $V=N$. For every set $S_j=\{x_1, \dots x_{|S_j|}\}$, we construct a cycle between its elements where all edges have time label $j$, \ie for each $0 < i < |S_j|$ we add $\tdirect{x_i}{x_{i+1}}{j}$, and lastly $\tdirect{x_{|S_j|}}{x_1}{j}$. We ask whether \gcal admits a non-strict walk \eecShort using $k$ walks. As $k$-\hittingset is \Wtwo-hard when parameterized by $k$, \TWEECnoK{NS-} on directed graphs is also \Wtwo-hard. See \Cref{fig:NSD_TWD_proof} for an example.
    \begin{figure}[ht]
        \centering
        \includegraphics[width=0.85\columnwidth]{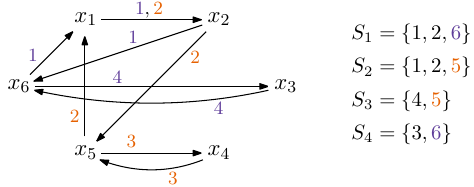}
        \caption{Example of the construction for \Cref{thm:DNS_TWD}.
        }
        \label{fig:NSD_TWD_proof}
    \end{figure}
    
    \bigparagraph{($\Rightarrow$)\quad} 
    Given a hitting set $H=\{h_1, \dots h_k\}$ for the instance $(N, S, k)$, we iteratively construct an exact edge-cover of the same size. First, we start one walk on every $h_i \in H$ at time step $0$. Then, for each time step $0<j\leq m$ corresponding to a set $S_j$, pick any vertex $h\in S_j \cap H$ and extend the walk currently ending on this vertex to cover all edges with time label $j$.
    Since those edges form a cycle through all elements of $S_j$, including $h$, the walk can be extended in this way and will loop back to $h$. Therefore, at the next time step, this walk can be extended from $h$ again.
    Every edge in \gcal at time step $j$ corresponds to some set $S_j$. Since $H$ is a hitting set, for each set $S_j$ we know that $S_j\cap H \neq \emptyset$ and therefore every edge with time step $j$ will be covered exactly once by a walk looping through a vertex $h\in S_j \cap H$.
    
    \bigparagraph{($\Leftarrow$)\quad} 
    Given an exact edge-cover \wcal for \gcal, we show that the set of start-terminals of all walks is a solution to the instance of \hittingset.
    First, we show that after every time step $j$, the set of end-terminals of the sub-walks between time-step 1 and $j$ is the same as before the time step:
    After time step 0, the set of start-terminals forms the set of current end-terminals.
    For every time step $j>0$, the edges in $E_j$ form a cycle by construction. If a sub-walk begins on a vertex and uses the only outgoing edge, there must be a walk that covers the incoming edge and ends at this vertex for lack of available outgoing edges. Therefore, the set of end-terminals remains unchanged.
    It follows, by induction, that the set of start-terminals $H$ is exactly the set of end-terminals after each time step. Since \wcal covers all edges of each $E_j$, at least one walk must end at a vertex contained in $S_j$ after time step $j-1$. As mentioned before, this vertex is also a start-terminal and thus part of our hitting set $H$.
    \end{proof}
    \fi
        
    Moving on to undirected graphs, a walk must choose the direction of its starting edge.
    Using this choice as an assignment, we can construct a reduction from \sat{3} \iflong to \TWEECnoK{(N)S-} on undirected graphs\fi.
    \ifshort
    The strict and non-strict variants require slightly different constructions but share the same idea: Each variable corresponds to an undirected edge at time step 1. Traversing that edge from left to right sets the variable to True, while traversing it from right to left, sets the variable to False.
    \begin{theorem}[$\star$]
    \label{thm:NSUD_TWD}
        \TWEECnoK{(N)S-} on undirected temporal graphs is \NP-complete. 
    \end{theorem}
    \fi
    
    \iflong 
    \begin{theorem}
    \label{thm:SUD_TWD}
        \TWEECnoK{S-} on undirected temporal graphs is \NP-hard. 
    \end{theorem}
    \begin{proof}
    We reduce from \sat{3}. 
    Let $\varphi = \bigwedge_{i\in[\ell]}(\bigvee_{j\in[3]}L_{i,j})$ be a \sat{3} instance with $h$ many variables and $\ell$ many clauses. We construct an undirected temporal graph $\gcal:= \tuple{G, \ecal}$. 
    For each variable $X_j$ in $\varphi$, we create a \emph{variable-gadget} containing 
    two connected vertices $T_j$ and $F_j$, corresponding to the True/False directions out of the gadget, and the edge $\tundirect{T_j}{F_j}{0}$.
    For each clause $C_i$ in $\varphi$, we create a \emph{clause-gadget} containing 
    $5$ vertices $c_i$, $a_i^1 , a_i^2, a_i^3, a_i^4$ and the following temporal edges: \emph{start-cleanup-edges} $\tundirect{c_i}{a_i^1}{0}$, $\tundirect{c_i}{a_i^2}{0}$ and \emph{end-cleanup-edges} $\tundirect{c_i}{a_i^3}{3i+3}$, $\tundirect{c_i}{a_i^4}{3i+3}$.
    
    Now, we connect the clause-gadgets and variable-gadgets using \emph{assignment-edge-pairs}. For each variable $X_j$ in $\varphi$, if $X_j$ appears as a positive literal in the clause $C_i$, we add \tundirect{c_i}{T_j}{3i+1} and \tundirect{c_i}{T_j}{3i+2}. If $X_j$ appears as a negative literal in the clause $C_i$, we add \tundirect{c_i}{F_j}{3i+1} and \tundirect{c_i} {F_j}{3i+2}. 
    This completes the construction.
    See Figure \ref{fig:SUD_TWD_proof} for an example.
    
    \begin{figure}[t]
            \centering
            \includegraphics[width=\columnwidth]{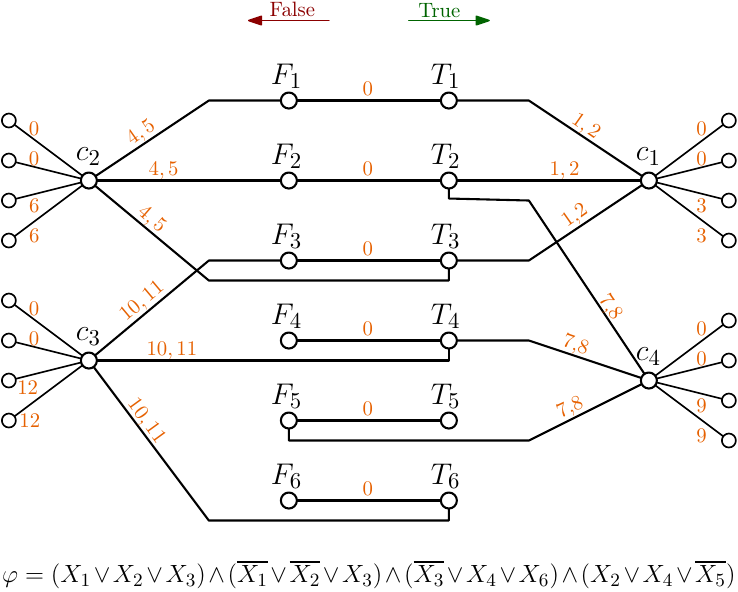}
            \caption{Example of the construction for \Cref{thm:SUD_TWD}. 
            }
            \label{fig:SUD_TWD_proof}
        \end{figure}
    
    First, observe that the size of any EEC for \gcal is at least $h+2\ell$ because \gcal contains exactly this many edges with time label $0$ and we only allow strict walks.
    We show that there is a satisfying assignment for $\varphi$ if and only if $\gcal$ allows for a strict \eecShort with exactly $h+2\ell$ walks.
    
    The main idea of the reduction is to utilize the undirected edges to simulate the assignment of True or False to a variable. A walk taking a variable-edge from \textit{left to right} will be equivalent to that variable being set to True and, respectively, a walk going from \emph{right to left} will be equivalent to setting that variable to False. A satisfying assignment will then enable the walks - forced to form by the edges with time step zero - to cover all assignment-edges without needing an additional path. 
    
    In the following we will emphasize the direction in which an undirected edge is taken by the order in the set, \ie a nonempty strict temporal walk taking the edge \tundirect{v}{u}{t} can only be extended after time step $t$ and only at the vertex $u$, while after taking \tundirect{u}{v}{t}, the walk can only be extended at $v$.

    \bigparagraph{($\Rightarrow$)\quad}
        Let $\beta$ be a satisfying assignment for $\varphi$.
        We define a set $\wcal$ of temporal walks as follows.
        For every variable $X_j$ in $\varphi$, we start a walk $w_{X_j}$ at the variable edge $\tundirect{F_j}{T_j}{0}$, \ie from \emph{left to right}, if and only if $\beta(X_j)=1$. 
        For every clause $C_i$ in $\varphi$, we start two walks $w_{c_i,1}, w_{c_i,2}$ at the start-cleanup-edges $\tundirect{a_i^1}{c_i}{0}$ and $\tundirect{a_i^2}{c_i}{0}$ towards the vertex $c_i$.
        
        Now, to cover the assignment-edge-pairs connecting a clause to the gadgets of its contained variables, we extend these two types of walks. 
        First, note that since $\beta$ is a satisfying assignment, at least one literal of each clause has to be satisfied; without loss of generality, we assume that to be the first literal.
        Now, for every clause $C_i$, let $X_1,X_2,X_3$ be the variables corresponding to the literals $L_{i,1}, L_{i,2}, L_{i,3}$ in $C_i$, respectively. We extend $w_{X_1}$ by the assignment-edge-pair $\tundirect{T_1}{c_i}{3i+1}, \tundirect{c_i}{T_1}{3i+2}$, if $L_{i,1}=X_1$ is a positive literal. Otherwise, if $L_{i,1}=\overline{X}_1$ is a negative literal, we extend $w_{X_1}$ by the edge-pair $\tundirect{F_1}{c_i}{3i+1}, \tundirect{c_i}{F_1}{3i+2}$. This covers the first assignment-edge-pair adjacent to $c_i$.
        To cover the other two edge-pairs, we extend $w_{c_i,1}$ and $w_{c_i,2}$, and, after that, also add the end-cleanup-edges adjacent to $c_i$.
        So, for example if $L_{i,2}=X_2$ is a positive literal and $L_{i,3}=\overline{X_3}$ is a negative literal, the walks would be constructed as $w_{c_i,1}=(\tundirect{a_i^1}{c_i}{0}, \tundirect{c_i}{T_2}{3i+1}, \tundirect{T_2}{c_i}{3i+2}, \tundirect{c_i}{a_i^3}{3i+3})$ and $w_{c_i,2}=(\tundirect{a_i^2}{c_i}{0}, \tundirect{c_i}{F_3}{3i+1}, \tundirect{F_3}{c_i}{3i+2}, \tundirect{c_i}{a_i^3}{3i+3})$.
        
        By construction, the walks in \wcal are properly defined, every temporal edge is contained exactly once and they form walks adhering to the temporal order of edges. Observe that the last claim follows from $\beta$ being a satisfying assignment in conjunction with our definition of the walks $w_{X_j}$. Because, the walk takes the variable-edge towards the True-vertex if and only if $\beta(X_j)=1$, $w_{X_j}$ is on the ``correct'' side of the edge to cover the assignment-edge-pair whenever $X_j$ is in the first literal -- which, by our assumption, is satisfied by $\beta$.
        As $\lvert \wcal \rvert = h+2\ell$ and by our earlier observation, any minimal strict walk exact edge-cover in \gcal has at least $h+2\ell$ walks, \wcal is minimal.
    
    \bigparagraph{($\Leftarrow$)\quad}
        For the other direction, let $\wcal$ be an exact edge-cover with strict walks in \gcal of size $h+2\ell$. 
        We define a truth-value assignment $\beta$ for $\varphi$ by setting $\beta(X_j)=1$ if and only if the edge in the variable-gadget of $X$ in \gcal is taken from \emph{left to right}, \ie as $((F_j,T_j),0)$. As every edge is contained in exactly one walk of \wcal, this assignment is properly defined.
        We show that $\beta$ satisfies the \sat{3} formula $\varphi$.
    
        First, observe that there has to be a distinct walk for every edge at time step 0, \ie one walk for each variable-gadget and two walks starting in each clause-gadget. As there are exactly $h+2\ell$ such edges, there can be no additional walks in \wcal.
        Additionally, in the $i$\textsuperscript{th} clause there are two edges at time step $3i+3$ adjacent to $c_i$ happening after any other edge adjacent to $c_i$. At these edges, two walks have to end.
        
        The edges whose membership to the walks is mostly unclear are the assignment-edge-pairs.
        Without loss of generality, we assume that those edge-pairs are both contained in the same walk, as the set of end-points for the walks after the corresponding time-steps would be the same -- one walk would be on the side of the variable-gadget and the other walk would be on the side of the clause-gadget. 
    
        Out of the three assignment-edge-pairs adjacent to one clause, only up to two can be covered by the walks starting in the clause-gadget without needing an additional walk. Additionally, one walk can only cover at most one edge-pair as the walks are strict and the three edge-pairs have the same time labels.
        This implies that at least one of these pairs will have to be covered by a different walk. The only option to cover such an edge-pair without introducing a new walk is to extend a walk from the \emph{correct} side of the variable-gadget, \ie the walk has to be on the True-side if the literal was positive, otherwise on the False-side.
        Observe that after covering an assignment-edge-pair, the walk that started in the variable-gadget will be on the same side in the variable gadget again and since the edges adjacent to one clause-gadget happen all before the next clause-gadget, that walk corresponding to the variable can cover assignment-edge-pairs of as many clauses as possible. 

        In summary, there is one walk starting at every edge with time step 0. In particular, there is one walk for each variable in $\varphi$, ending on the True or False side of each variable-gadget. Furthermore, for each clause-gadget there is one assignment-edge-pair connecting to the corresponding side of the variable-gadget of one of the three contained literals, which is not covered by a walk starting in the clause-gadget, but has to be covered by a walk starting in the variable-gadget. This can be if and only if the variable-edge was taken \emph{left to right} if the literal was positive and \emph{right to left} if the literal was negative.
        Therefore, by setting $\beta(X)=1$ if and only if the variable-edge was taken \emph{left to right}, $\beta$ satisfies every clause in $\varphi$.
    \end{proof}
    \begin{theorem}
    \label{thm:NSUD_TWD}
        \TWEECnoK{NS-} on undirected temporal graphs is \NP-hard.
    \end{theorem}
    \begin{proof} We reduce from \sat{3}.
    Given a \sat{3} instance $\varphi = \bigwedge_{i\in[\ell]}(\bigvee_{j\in[3]}L_{i,j})$, we construct an undirected temporal graph $\gcal:= \tuple{G, \ecal}$ similar to the construction in \Cref{thm:SUD_TWD} for strict walks.
    The construction of the variable-gadgets is exactly the same.
    And for each clause $C_i$ in $\varphi$, we create a clause vertex $c_i$.

    \begin{figure}[ht]
        \centering
        \includegraphics[width=\columnwidth]{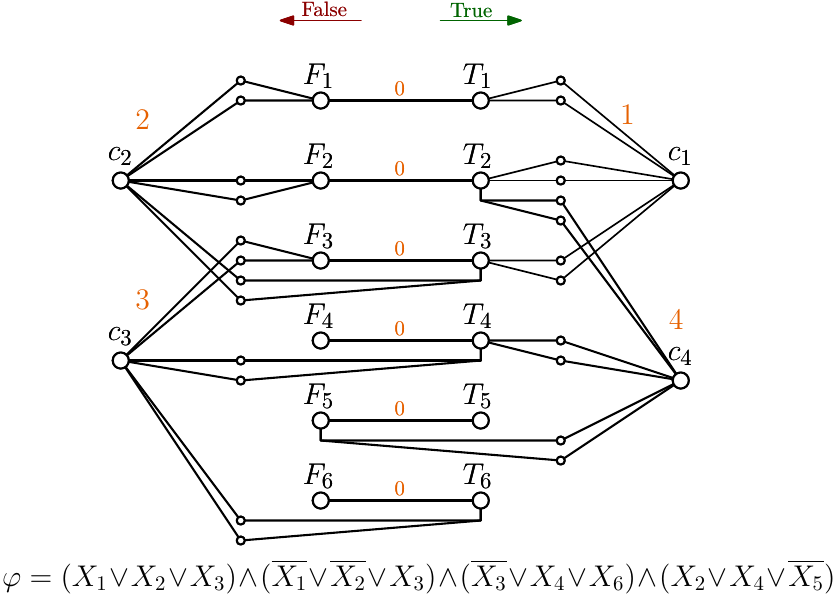}
        \caption{Example of the  construction for \Cref{thm:NSUD_TWD}. 
        }
        \label{fig:NSUD_TWD_proof}
    \end{figure}
    The main difference lies in how we construct the assignment-edges.
    Because a non-strict walk starting in a variable-edge can consecutively take edges with the same time label, we can dedicate to each clause $C_i$ a time label $i$ and use this label on all assignment-edges adjacent to the clause.
    For each literal $L_{i,j}$ in clause $C_i$ we construct two edges connecting the clause-vertex with the corresponding side of the variable-gadget at time step $i$. As this would result in a multigraph, we require one intermediate vertex on each edge. Refer to \Cref{fig:NSUD_TWD_proof} for an illustration.
    For example, for a positive literal $L_{i,1}=X_1$ contained in $C_i$, we add:
    \tundirect{c_i}{b_i^{1}}{i} , \tundirect{b_i^{1}}{T_1}{i} and \tundirect{c_i}{b_i^{2}}{i} ,\tundirect{b_i^{2}}{T_1}{i}.
    Observe that there are $h$ non-adjacent zero edges, so any exact edge-cover contains at least $h$ walks.

    The main idea is very similar to that of \Cref{thm:SUD_TWD}; the key difference is how we cover the assignment-edges without the cleanup walks.
    As previously, we have $h$ walks starting at variable-edges, going left/right to simulate True/False assignments. Each of these walks will extend to cover \emph{all} assignment-edges for every clause $C_i$ in which they are the first literal $L_{i,1}$ (satisfying the clause). 
    We can do this with exactly one walk because they are \emph{non-strict}.
    So, utilizing that assignment-edges associated with a clause-gadget all have the same time label and are connected, the walk of the satisfying variable $w_{X_1}$ will cover all assignment-edge-pairs, in contrast to covering just its own edge-pair.
    \end{proof}
    \fi

\section{Walks with Fixed Terminals -- Polynomial Time Algorithms}
\label{sec:fxdTerminals}
    In this section, we focus on exact edge-covers with fixed terminals.
    Observe that this also fixes the number of journeys (one journey per terminal pair), but we emphasize that this is \textbf{not} the same as having a \textbf{constant} number of journeys.
    
    It is easy to see that our hardness constructions for paths and trails inherently fix the terminals, thereby translating directly.
    In contrast, for walks, we \textit{chose} the terminals to \textit{simulate} a Boolean assignment. This is not incidental, 
    as having fixed terminals makes finding an eec with walks tractable.

    \subsection{Directed Graphs}
    
    Strict walks in directed graphs can be solved without fixing terminals using a dynamic program computing possible endpoints with increasing time. For non-strict walks with fixed terminals, this program can be adapted, as we avoid the issue of choosing the vertices to start on (in particular on a cycle). 
    So, we initiate one walk per start-terminal and extend those at each time step, possibly by multiple edges.
    \iflong 
    \begin{theorem}
    \fi
    \ifshort
    \begin{theorem}[$\star$]
    \fi\label{thm:DNS_TWD_fixed_terminals}
        \TWEECnoK{NS-} on directed temporal graphs can be solved in polynomial time if the start-terminals along with their multiplicity are part of the input.
    \end{theorem}
    \iflong
    
\begin{proof}
    \newcommand{\TermStartSet}[1]{\ensuremath{S_{#1}}\xspace}
    \newcommand{\TermEndSet}[1]{\ensuremath{Z_{#1}}\xspace}
    Let $\gcal$ be a directed temporal graph and $V_S,V_E\subseteq V$ the multiset of $k$ start-terminals and end-terminals, respectively. 
    We provide a constructive proof that iteratively builds exact edge-covers $\wcal_t$ for $\gcal_{\leq t}$, where $\gcal_{\leq t}$ denotes the temporal subgraph of \gcal which contains only edges up to time label $t$.
    For each timestep $t$, we store in the multiset $V_t$ the vertices that each walk is currently ``sitting'' in. So, if 3 walks arrive in $u$ before timestep $t$ and do not leave $u$ before $t+1$, then $\{u,u,u\}\in V_t$.
    We initialize the multiset as $V_0=V_S$ and $\wcal_0$ with empty walks, one for each vertex in $V_S$.

    At time step $t$, the algorithm extends the walks in $\wcal_{t-1}$ to cover all edges $E_t$ in the snapshot $\gcal_t$.
    It does so by iteratively identifying cycles in $\gcal_t$ and attaching those to any walk that is on the cycle. This does not change $V_t$. If there is no walk on any cycle, the algorithm returns there is no walk eec. 
    Now, assume there is no cycle in $\gcal_t$. Then every edge $((u,v),t)$ is attached to any walk in $u$ and we update $V_t=V_t-\{u\}+\{v\}$. If the algorithm processes an edge for which $u\notin V_t$, the algorithm returns there is no walk eec.  
    By construction, $\wcal_t$ is an exact edge-cover for $\gcal_{\leq t}$.

    First, we show that $V_{t+1}$ is unique.
    This implies that if the algorithm outputs a $\wcal_{\tmax}$ after processing every time step, $\wcal_{\tmax}$ is an exact edge-cover for $\gcal$ if and only if $V_{\tmax}=V_E$, if one exists.

    Assume we process $t$ with the walks $\wcal_{t-1}$ sitting in $V_t$. We denote by $s_t(v)$ the number of walks sitting in $v$.
    For each vertex $v$, we denote the offset between the number of incoming and outgoing temporal edges for each vertex $v$ by $o_t(v) = \tindegree_t(v) + s_t(v) - \toutdegree_t(v)$. Recall, $\tindegree_t(v)$ and $\toutdegree_t(v)$ denote the temporal edges of $E_t$ adjacent to $v$ and $s_t(v)$ denotes the number of walks starting on vertex $v$ at $t$.

    Two conditions must be met to cover all edges of $E_t$ exactly once: 
    (1) There can be no vertex $v$ with $o_t(v) < 0$, and (2) all edges in $E_t$ must be reachable from vertices in $V_t$ by using only edges in $E_t$.
    If (1) is violated by a vertex, it can not be visited by walks often enough to cover all outgoing edges.
    Condition (2) must hold because we cannot extend walks from vertices that are not in $V_t$.
    
    If (1) and (2) both hold, covering all edges in $E_t$ via walks extended from $V_t$ is possible.
    All that's left is to show that the multiset of end-vertices is the same for every possible cover.
    The value $o_t(v)\geq 0$ denotes the minimum number of walks sitting in a vertex $v$ after time step $t$, as it is the difference between incoming (plus start-terminals) and outgoing edges.
    Since an outgoing edge for one vertex is an incoming edge for another, it holds $\sum_{v \in V} s_t(v)=\sum_{v\in V}o_t(v)$.
    Hence, no additional end-vertices can exist, and $o_t(v)$ is exactly the number of walk sitting at $v$, \ie $s_t(v)=o_t(v)$. See \Cref{fig:NSD_TWD_fixed_terminals} for an illustration of the correspondence between the start- and end-vertices in each time step.
    \iflong 
    \begin{figure}[h]
        \centering
        \includegraphics[width=0.7\columnwidth]{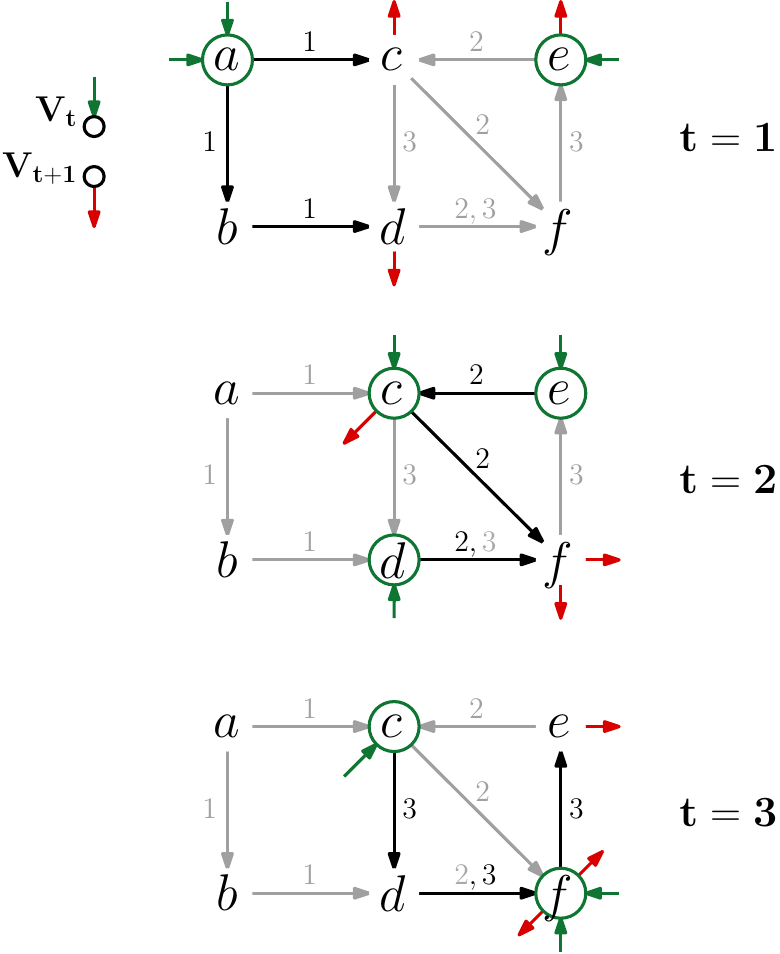}
        \caption{Illustration for non-strict walks in directed graphs \Cref{thm:DNS_TWD_fixed_terminals}. 
        For each time step $t\in[3]$, all but $E_t$ is grayed out, $V_t$ is illustrated by incoming green edges with no source and $V_{t+1}$ by outgoing red edges with no sink.}
        \label{fig:NSD_TWD_fixed_terminals}
    \end{figure}
    \fi
\end{proof}
    \fi

    \subsection{Undirected Graphs}
    In undirected graphs, even with terminals, we must choose the direction of each temporal edge connecting two walks.
    This resembles the problem of finding a circulation flow.
  
    \ifshort
    For non-strict walks, this can be simulated in a static mixed multigraph, similar to the static expansion by \citet{kostakos_temporal_2009}, and solved via  \SWEC with fixed terminals in polynomial time using a flow algorithm.
    \begin{theorem}[$\star$]
    \label{thm:static-walk-fxd-term}
        \SWEC with fixed terminals on static mixed multigraphs (containing directed and undirected multiedges) can be computed in polynomial time. 
    \end{theorem}\fi
    
    \ifshort
    
    Let us now state the main theorem of this section.
    \begin{theorem}[$\star$]
    \label{thm:UD_NS_TWD_fixed_terminals}
        \TWEECnoK{NS-} with fixed terminals on undirected temporal graphs can be computed in polynomial time.
    \end{theorem}

    \input{Proofs/proof_fixed-terminals_UD_(N)S}
    \fi
    
    \iflong 
\begin{figure*}[ht]
    \centering
    \includegraphics[width=\textwidth]{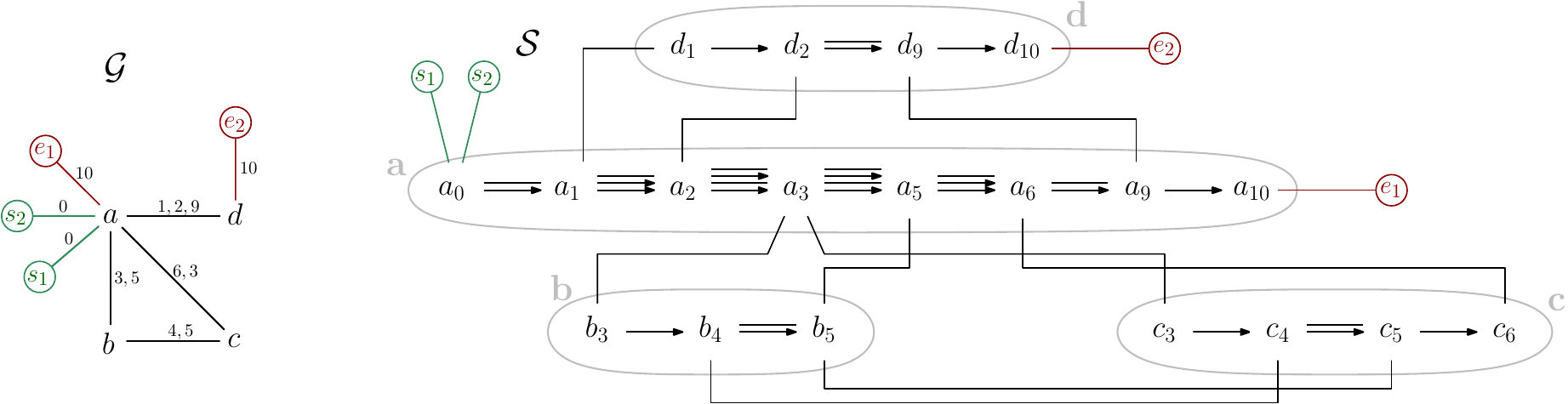}
    \caption{Illustration of how a temporal graph $\gcal$ (left) is translated into a static graph \scal (right) for \Cref{thm:UD_NS_TWD_fixed_terminals}. The terminals are $V_S=\{s_1, s_2\}$ and $V_E=\{e_1, e_2\}$. The vertices $\{a,b,c,d\}$ are replaced with gadgets based on their temporal degree.} 
    \label{fig:UD_TWD_temporal_to_static}
\end{figure*}
First, observe the following straightforward transformation of a temporal graph $\gcal$ with start- and end-terminals $V_S$ and $V_E$:  For each start-, respectively end-, terminal at $v$, we add an additional vertex $s_i$, resp. $e_i$, (with $i$ counting up from 1) along with an edge between $s_i$, resp. $e_i$, and $v$.
The edges to the start-terminals get a time label earlier than any label in \gcal, and the edges to the end-terminals a label later than any label in \gcal.
This yields a temporal graph $\gcal'$ with vertices $V\cup\TerminalStartSet\cup\TerminalEndSet$, whose walk exact edge-covers are equivalent to the eecs of \gcal. 
Note that a vertex in \gcal can be a start and end of multiple walks while at the terminal-vertices in $\gcal'$ there will be exactly one walk starting or ending.
In the following, we assume that we are given the temporal graph in the aforementioned form and
whenever we mention ``vertices'' this refers only to $V$ excluding the terminals.

The following theorem states our result.
\begin{theorem}
\label{thm:UD_NS_TWD_fixed_terminals}
    \TWEECnoK{NS-} for undirected temporal graphs can be solved in polynomial time if the start- and end-terminals with their multiplicity are part of the input.
\end{theorem}
We are going to prove the theorem through a construction into a static graph followed by a sequence of lemmas.
Given an undirected temporal graph \gcal with terminals \TerminalStartSet,\TerminalEndSet, we transform \gcal into a static multigraph \scal containing directed and undirected edges.
We give the construction in \Cref{def:fxdwalks-construction} and show in \Cref{lem:fxdwalks-temporalgcalTOstaticscal,lem:fxdwalks-staticscalTOtemporalgcal} that a temporal walk eec of \gcal corresponds to a static walk eec in \scal.
Then, in \Cref{lem:fxdwalks-staticscalTOstaticscal'}, we show how we can turn the mixed multigraph \scal into a purely undirected static graph $\scal'$, such that a static walk eec of \scal corresponds to a static walk eec in $\scal'$.
Finally, we show how to obtain a static walk eec in $\scal'$ by providing an algorithm for \SWEC in undirected graphs with given terminals and proving its correctness in \Cref{lem:fxdwalks-staticeec}.

The following definition describes the central construction. For an illustration of the construction on an example graph, refer to \Cref{fig:UD_TWD_temporal_to_static}.
\begin{definition} \label{def:fxdwalks-construction}
    Let $\gcal=\tuple{(V\cup\TerminalStartSet\cup\TerminalEndSet,E),\ecal}$ be a temporal graph with start- and end-terminals $\TerminalStartSet,\TerminalEndSet$ and for every $v\in V$ let $T(v)=\{t_1,\dots,t_{\ell(v)}\}$ be the multiset of labels on edges incident to $v$ in \gcal.
    Observe that without loss of generality the temporal degree of all vertices in $V$ must be even, since any time a walk of an eec passes through a vertex it has to use two edges.
    
    We construct a static mixed multigraph \scal with vertices $\{v_t \colon v\in V \text{ and } t\in T(v)\}\cup\TerminalStartSet\cup\TerminalEndSet$.
    We say the $v_t$ for each $v$ are placed from \emph{left to right} sorted by increasing time label, meaning $v_t$ is to the \emph{left} of $v_{t'}$ if and only if $t < t'$.

    \scal contains undirected edges between terminals and vertices as present in $\gcal$, \ie $\{\{s_i,v_0\} \colon (\{s_i,v\},0)\in\ecal\} \cup \{\{v_{\tmax},e_i\} \colon (\{v,e_i\},\tmax)\in\ecal\}$. Furthermore, any undirected edge between vertices $(\{u,v\},t)$ in \gcal translates to an undirected edge $\{u_t,v_t\}$ in \scal.
    %
    
    Lastly, \scal contains a specific number of directed and undirected multiedges connecting any $v_{t_i}$ with its right neighbor. 
    Let $\tdegree_{\leq t_i}(v)$ and $\tdegree_{> t_i}(v)$ be the number of temporal edges incident to $v$ in \gcal with time label $\leq t_i$ and $>t_i$, respectively.
    We create a total of $\alpha_i=\min(\tdegree_{\leq t_i}(v), \tdegree_{> t_i}(v))$ edges between $v_{t_i}$ and $v_{t_{i+1}}$ as follows: $\lceil\alpha_i/2\rceil$ are directed from $v_t$ to $v_{t+1}$ and the remaining $\lfloor\alpha_i/2\rfloor$ are undirected. We call this the \emph{$v$-gadget}.
\end{definition}

    From hereon, we assume \scal to be the static mixed multigraph from \Cref{def:fxdwalks-construction}.
    We now show that any temporal walk exact edge-cover in \gcal corresponds to a static walk exact edge-cover in \scal and vice versa.

    \begin{lemma} \label{lem:fxdwalks-temporalgcalTOstaticscal}
        A temporal walk exact edge-cover in \gcal corresponds to a static walk exact edge-cover in \scal.
    \end{lemma}
    \newcommand{\onevertex}{\ensuremath{u}} 
    \newcommand{\ellvertex}{\ensuremath{v}} 
    \begin{proof}
        Let $\mathcal{W}_\gcal$ be a temporal walk exact edge-cover of \gcal.
        Then any walk $W\in\mathcal{W}_\gcal$ has to be of the form $W=(\{s_i,\onevertex\},0)\circ  W[\onevertex,\ellvertex]\circ (\{\ellvertex,e_j\},\tmax)$ with the subwalk $W[\onevertex,\ellvertex]$ visiting only vertices in $V$.
        For each such $W\in\mathcal{W}_\gcal$, we create a static walk $W_\scal\in \mathcal{W}_\scal$ with $W_\scal = \{s_i,\onevertex_0\}\circ W_\scal[\onevertex_0,\ellvertex_{\tmax}]\circ \{\ellvertex_{\tmax},e_j\}$ with 
        the subwalk $W_\scal[\onevertex_i,\ellvertex_{\tmax}]$ defined 
        as follows:

        Let $W[\onevertex,\ellvertex]$ contain the two edges $(\{a, v\}, t_\alpha)$ and $(\{v, b\}, t_\beta)$ in succession.
        Then we add to $W_\scal[\onevertex_i,\ellvertex_{\tmax+j}]$ the path $a_{t_\alpha},v_{t_\alpha}, v_{t_{\alpha+1}}, \dots, v_{t_\beta}, b_{t_\beta}$. 
        So, $W_\scal$ enters the $v$-gadget via $v_{t_\alpha}$, continues to the right 
        and exits the gadget upon reaching $v_{t_\beta}$.
        This path uses directed edges where possible and only uses undirected edges if all directed edges have already been used. 

        To show that this definition is well defined, we have to ensure that there are enough edges between neighboring vertices $v_{t_i}$ and $v_{t_{i+1}}$ such that each edge is used at most once. 
        Note that at most $\delta_{\leq t_i}(v)$ walks can enter $v_{t_i}$ in \scal; either by entering to the left of it or at $v_{t_i}$ directly. By definition, this number corresponds to the maximum number of walks entering $v$ in \gcal until time step $t_i$.
        Further observe that since $\mathcal{W}_\gcal$ is an eec, any walk that leaves $v$ after time step $t_i$ uses a distinct time edge. So there can be at most $\delta_{>t_i}(v)$ many walks leaving $v_{t_i}$.
        Since we constructed a number of edges between $v_{t_i}$ and $v_{t_{i+1}}$ equaling the minimum out of those two terms, the number of edges is sufficient for the static walks.
    
        Now, we have to ensure that every edge is used at least once by the constructed walks in $\mathcal{W}_\scal$.
        With the current construction, all edges incident to terminals or between different vertices $u,v$ in $\gcal$ are covered exactly once but there might be some edges between some $v_{t_i}$, $v_{t_{i+1}}$ unused.
        First, observe that the number of edges remaining between any $v_{t_i}$ and its right neighbor $v_{t_{i+1}}$ is even: Let w.l.o.g. $\delta_{\leq t_i}(v)\leq \delta_{>t_i}(v)$
        and let $x$ be the number of walks in $\mathcal{W}_\scal$ that \textbf{exit} the $v$-gadget to the \textbf{left} of or at $v_{t_i}$. 
        Each of these $x$ walks has to enter and also exit the $v$-gadget to the left of $v_{t_i}$, taking two non-gadget edges.
        This leads to a total of $2x$ walks not taking an edge from $v_{t_i}$ to $v_{t_{i+1}}$ compared to the maximum possible number of walks passing through, which is an even number.        
        Second, we claim that we can form a cycle out of this even number of edges at $v_{t_i}$ which we can insert into any static walk entering or exiting. 
        Recall that directed edges are used before undirected edges.
        In particular, at least one directed edge is used when $\delta_{\leq t_i}(v)$ is odd.
        Therefore, at least half of the $2x$ remaining edges are undirected.
        These edges form a cycle going back and forth between $v_{t_i}$ and $v_{t_{i+1}}$.
        
        In the end, we have constructed a set of walks starting at a distinct start-terminal out of $V_S$ and ending at a distinct end-terminal out of $V_E$ that are edge-disjoint and use every edge once. This is a static walk exact edge-cover for \scal with start- and end-terminals $V_S,V_E$.
    \end{proof}

    \begin{lemma}  \label{lem:fxdwalks-staticscalTOtemporalgcal}
        A static walk exact edge-cover in \scal corresponds to a temporal walk exact edge-cover in \gcal.
    \end{lemma}
    \begin{proof}
    Let $\mathcal{W}_\scal$ be a static walk eec for \scal.
    Note that one cannot directly translate $\mathcal{W}_\scal$:
    For a temporal walk eec, walks must exit a vertex via a time label greater than the label by which they entered the vertex.
    For the static graph \scal, this corresponds to (1) \textit{any walk has to enter a $v$-gadget to the left of where it exits}. However, this may be violated by a static walk eec due to the undirected edges in \scal.
    We show that the walks in $\mathcal{W}_\scal$ can be transformed such that property (1) holds. This is done by switching suffixes between overlapping walks.

    \newcommand{\inWone}{\ensuremath{{in}(W_1,v)}} 
    \newcommand{\outWone}{\ensuremath{{out}(W_1,v)}} 
    Let $W_1\in\mathcal{W}_\scal$ be an arbitrary walk violating (1), \ie it enters the $v$-gadget at $v_{t_\beta}$ and exits at $v_{t_\alpha}$ with $t_\beta > t_\alpha$.
    This implies that $W_1$ uses undirected edges from $v_{t_\beta}$ to $v_{t_\alpha}$.
    We now show how we switch suffixes to reduce the \textit{time travel} \textit{distance} 
    $t_\alpha-t_\beta$
    by at least one as long as $t_\alpha-t_\beta<0$ which ultimately leads to $t_\alpha-t_\beta\geq0$. 
    Refer to \Cref{fig:UD_TWD_redirected_static_walks} for an illustration.

    $W_1$ is taking an undirected edge $\{v_{t_{\alpha+1}},v_{t_{\alpha}}\}$. 
    By construction there are at least as many directed as undirected edges and every edge is part of a walk in the eec, so there has to be a directed edge $(v_{t_{\alpha}},v_{t_{\alpha+1}})$ which is contained in a walk $W_2\in\mathcal{W}_\scal$ entering the $v$-gadget at $t_{\alpha'}$ and exiting at $t_\gamma$ with $t_\alpha'\leq t_\alpha<t_\gamma$.
    We now switch the suffixes of $W_1$ and $W_2$ at $v_{t_\alpha}$: Let $s(W),e(W)$ denote the start- and end-terminal of a walk and $W[a:b]$ the subwalk of $W$ between $a$ and $b$. Then we define $W_1'=W_1[s(W_1):v_{t_{\alpha}}]\circ W_2[v_{t_{\alpha}},e(W_2)]$ and $W_2'=W_2[s(W_2):v_{t_{\alpha}}]\circ W_1[v_{t_{\alpha}},e(W_1)]$, where $\circ$ is the concatenation operator. Then $W_1$ exits the $v$-gadget at $v_{t_\gamma}$ which is right of $v_{t_\alpha}$, so the time travel distance decreased ($t_\gamma-t_\beta<t_\alpha-t_\beta$). Further note that since $t_\alpha'\leq t_\alpha$, $W_2'$ which enters the $v$-gadget at $t_\alpha'$ and exits at $t_\alpha$, still satisfies (1).
    So we reduced the backwards time travel distance of $W_1$ by at least one without breaking any other walks. 

    \begin{figure*}[ht]
        \centering
        \includegraphics[width=\textwidth]{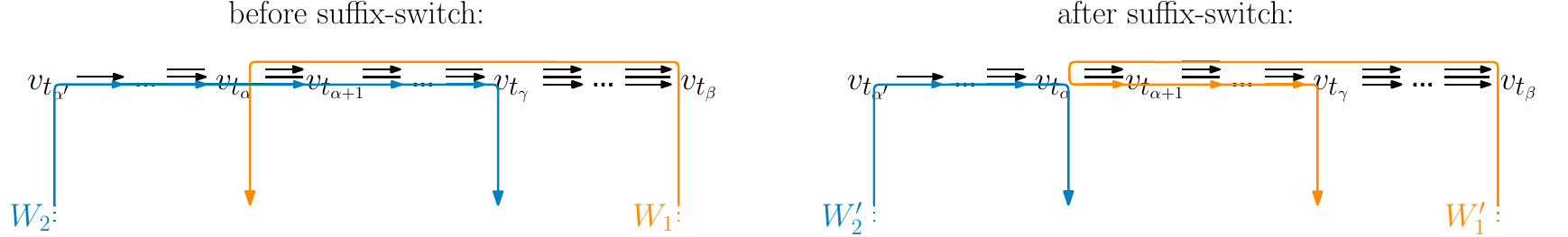}
        \caption{Illustration of how a walk $W_1$ is adjusted if it exits the $v$-gadget left of where it enters it. At $v_{t_\alpha}$ we swap the suffixes of $W_1$ and $W_2$, resulting in the new walks $W_1'$, $W_2'$.
        }
        \label{fig:UD_TWD_redirected_static_walks}
    \end{figure*}
    After redirecting every walk of $\mathcal{W}_\scal$ if necessary, we get a static walk eec $\mathcal{W}_\scal'$ which satisfies property (1).
    Now, we can directly translate $\mathcal{W}_\scal'$ into a temporal eec for \gcal by contracting the vertex-gadgets back into single vertices.
    \end{proof}

    Having proven that the problems on \gcal and \scal are equivalent, we proceed to show the transformation of \scal into an undirected static graph $\scal'$ with equivalent walk eecs and how to solve \SWEC for \scal using $\scal'$.
    \begin{lemma}
    \label{lem:fxdwalks-staticscalTOstaticscal'}
        A static walk exact edge-cover for \scal  can be found in polynomial time, if one exists.
    \end{lemma}
    \begin{proof}
    Let $(u,v)$ be a directed edge in \scal. Then a static walk eec has to contain exactly one walk that includes $(u,v)$. We can interpret this walk as two walks -- the prefix-walk before $(u,v)$ and the suffix-walk after $(u,v)$. 
    Therefore, we can remove $(u, v)$ and add a new end-terminal connected to $u$ and a new start-terminal connected to $v$. Using this process, we can transform \scal into an undirected graph $\scal'$.

    For $\scal'$ we can use \Cref{lem:static_walks_fixed_terminals} to find a static walk eec if it exists with $S=V_S$ and $T=V_E$.
    
    Given a solution for $\scal'$, we turn all the artificial start- and end-terminals back into directed edges and concatenate the corresponding walks.
    Note that this concatenation may lead to cycles that only contain terminals created from directed edges, and no true terminals.
    However, if such cycles exist, at least one of them must share a vertex $v$ with a walk that connects true terminals, because the graph is connected.
    This cycle can be inserted into the walk when it visits $v$.
    We can repeat this until no cycles are left, and we are left with a collection of walks going from true start-terminals to true end-terminals.
    \end{proof}

    For the proof of \Cref{lem:fxdwalks-staticscalTOstaticscal'}, need an algorithm solving \SWEC on an undirected graph with fixed terminals in polynomial time. That this is possible, we show below.

\begin{lemma}
\label{lem:fxdwalks-staticeec}
\label{lem:static_walks_fixed_terminals}
    Let $G$ be a connected undirected graph and $S,T\subseteq V(G)$ such that $S\cap T = \emptyset$, $|S|=|T|$, and every vertex in $S\cup T$ has degree one.
    There exists a static walk exact edge-cover using $|S|$ walks of $G$ such that each walk starts in $S$ and finishes in $T$ if and only if
    \begin{itemize}
        \item every vertex in $V(G)\setminus (S\cup T)$ has even degree, and
        \item the size of the minimum $S$-$T$ edge-cut is $|S|$. 
    \end{itemize}
    Moreover, if such an eec exists, it can be found in polynomial time.   
\end{lemma}
\begin{proof}
    Clearly, if such a static walk eec of size $|S|$ exists, there exist $|S|$ many edge-disjoint walks each starting in a vertex in $S$ and terminating in a vertex in $T$. It is well-known that if there is an $s$-$t$ walk $W$ in a graph $G$, then there exists $s$-$t$ path $P$ such that the edges of $P$ form a subset of the edges of $W$. It follows that there are $|S|$ many edge-disjoint paths from $S$ to $T$. Therefore, by Menger's theorem~\cite{menger1927allgemeinen,Diestel2012book}, the size of the minimum $S$-$T$ edge-cut is $|S|$. Moreover, for every vertex $v$ in $V(G)\setminus(S\cup T)$, each walk enters $v$ the same number of times as it leaves $v$ and each edge is visited by exactly one walk exactly once. Hence every vertex in $V(G)\setminus (S\cup T)$ has even degree.

    On the other hand, assume that every vertex in $V(G)\setminus (S\cup T)$ has even degree, 
    and the size of minimum $S$-$T$ edge-cut is $|S|$.
    By Menger's theorem, there exists $|S|$ edge-disjoint paths from a vertex in $S$ to a vertex in $T$ and we can find such set \pcal of paths for example by Ford-Fulkerson's algorithm~\cite{EdmondsKarp72flow_agorithm}.
    Since vertices in $S\cup T$ all have degree one, each of these paths starts in a different vertex of $S$ and terminates in a different vertex of $T$ and these paths cover all the edges incident on vertices in $S\cup T$.
    Note that in $G' = G- \bigcup_{P\in \pcal}E(P)$ that is obtained from $G$ by removing all the edges on some path in \pcal, every vertex in $V(G)\setminus (S\cup T)$ still has even degree.
    Moreover, since $G$ is connected, for every vertex $w$ that is not on any path in \pcal, exists a vertex $v$ such that there is a $w$-$v$-path in $G'$ and $v$ is on a path in \pcal.
    To obtain a static walk eec from \pcal we follow a very similar idea to Hierholzer's algorithm~\cite{hierholzer1873moglichkeit} for finding an Eulerian trail in a graph.
    
    We begin with a set of walks \wcal equal to \pcal and repeatedly augment them with missing edges.
    Every time, we pick a vertex $v$ that is on some walk $W$ in \wcal such that $v$ is still incident to at least two edges in $G'$. Find an Eulerian cycle $C$ that starts and ends on $v$ in the connected component of $G'$ that contains $v$ (for example using Hierholzer's algorithm). Let $W= W_{sv}\circ v \circ W_{vt}$, where $\circ$ is the concatenation operator. We replace $W$ in \wcal by the walk $W_{sv}\circ v \circ C\circ v\circ W_{vt}$, such that $W$ now makes a detour at $v$ to follow the Eulerian cycle $C$, and we remove $C$ from $G'$. We repeat until $G'$ is empty and all edges of $G$ are in some walk in \wcal.
\end{proof}


 \fi
    For strict walks, this construction fails as it allows exact edge-covers which are valid for non-strict walks but invalid for strict walks and, unfortunately, it cannot be directly modified to handle such cases.
    However, on temporal graphs where adjacent edges are never present at the same time, every strict walk is inherently non-strict, allowing the algorithm to work. 
    Such graphs are commonly referred to as proper temporal graphs \cite{casteigts_simple_2024,christiann2024inefficiently}.
    \begin{corollary} \label{cor:UD_S_fxdTerm_walks}
        \TWEECnoK{S-} with fixed terminals on proper undirected temporal graphs, \ie 
        $\lvert N_t(v)\rvert \in\{0,1\}$ for all $t\in[\tmax]$, can be computed in polynomial time, where $N_t(v)$ is the multiset of neighbours of $v$ at time step $t$. 
    \end{corollary}
    \ifshort
    Lastly, we cover the other extreme of undirected temporal graphs with exactly $k$ (or 0) edges per time step. Here, once again a dynamic program tracking possible walk endpoints is able to compute the eecs. This is possible because each walk must cover one edge per time step. \fi
    \iflong
    Additionally, we present a dynamic program for strict walks with fixed terminals on undirected graphs with exactly $k$ (or 0) edges per time step. The algorithm keeps track of possible current endpoints, as in \Cref{prop:walks_XP}. 
    When there are exactly $k$ edges at a time, every walk has to cover one edge and the next configuration from given current endpoints is unique. 
    \fi
    \iflong 
    \begin{theorem}
    \fi
    \ifshort
    \begin{theorem}[$\star$]
    \fi
        \TWEECnoK{S-} with fixed terminals $V_S$, $V_E$ on undirected temporal graphs where $\lvert E_t\rvert\in\{0,k\}$ for all $t\in[\tmax]$, can be computed in time $\mathcal{O}(2^{\lvert V_S\rvert}\cdot\lvert V_S\rvert\cdot\tmax)$.
    \end{theorem}
    \iflong\begin{proof}
        Let \gcal be a temporal graph with exactly 0 or $k$ edges at each time step, and let $V_S$, $V_E$ be multisets of vertices in \gcal denoting the start- and end-terminals, respectively. Since the walks of the exact edge-cover are  strict, each walk must cover exactly one edge in each time step with $k$ edges.

        We denote by $\walks(v,t)$ the number of walks at vertex $v$ after the edges at time $t$ have been traversed. Initially, $\walks(v,0)$ is set to the number of occurrences of $v$ in $V_S$, \ie the number of walks starting in $v$.
        We show that the next configuration is unique.

        Assume the algorithm processes $k$ edges at time step $t$. 
        Observe that if a walk eec with the given terminals exists, all walks must leave their current vertices; otherwise, they cannot cover any edge. Since the number of walks equals the number of edges to be covered, it follows that $\tdegree_t(v)\geq\walks(v,t-1)$, and since any edge not taken by a walk from $v$ is taken towards $v$, we have $\walks(v,t)=\tdegree_t(v)-\walks(v,t-1)$.
        This computation is performed for every vertex $v$, yielding a unique configuration for $t$.

        Now, we check whether such a configuration is actually possible. As there are $k$ edges, there are exactly $2^k$ ways to orient the edges. Each such orientation requires a unique configuration of walks in order to be traversable (one walk at the source of each directed edge in the orientation). If one of the $2^k$ orientations fits the configuration encoded by $\walks(v,t-1)$ for all $v$, then the $k$ edges at time step $t$ can be covered by the walks starting in $V_S$ and the configuration encoded by $\walks(v,t)=\tdegree_t(v)-\walks(v,t-1)$ is valid.
        Otherwise, if there is no such orientation, the $k$ edges at time step $t$ cannot be covered by the walks starting in $V_S$ and there exists no walk eec with the given terminals. 

        After processing all temporal edges, we check whether $\walks(v,\tmax)$ equals the number of occurrences of $v$ in $V_E$, \ie the number of walks ending at $v$.
        If so, the entries in $\walks(v,t)$ encode a strict walk eec; otherwise, no strict walk eec exists for the given terminals.
    \end{proof}\fi
    For temporal graphs with an arbitrary number of edges per time step, the complexity of  \TWEECnoK{S-} remains an open question.

\section{Conclusion}
\label{sec:discussion}
    We introduce exact edge-covers (eecs) on temporal graphs with three journey variants: paths, trails and walks, aiming to capture the subclass of temporal graphs that model real-world transit networks.
    We provide a comprehensive analysis of the computational complexity of the eec problem and observe fundamental differences between covering with paths/trails and covering with walks.
    
    Several open problems remain, motivated by both theory and practical considerations.
    On the theoretical side, it would be interesting to explore whether walk eecs can be approximated and to resolve the open case for strict walks on undirected graphs with fixed terminals.
    Additionally, for paths, we observed that eecs are efficiently computable on DAGs and bidirected trees.
    So, an open question is to classify the subclasses of graphs for which these generally hard problems are polynomial-time solvable.
    Designing efficient algorithms for paths and trails with parameters arising from real-life applications is also an important direction.

    On the applied side, several variants of the model could be explored.
    Our focus was on the exact covering version of the problem, as it captures scenarios where a connection cannot be used by multiple vehicles at the same time, such as railroads, where only one train can occupy a track at a time.
    Additionally, studying the exact version implicitly minimizes the resources, like fuel, needed to satisfy the connections.
    However, the non-exact case is also interesting to explore. While our negative results still apply, the algorithms would need significant adjustments or entirely new approaches.

    A completely different, and ``complementary'', direction is to consider that the underlying temporal graph is {\em provided} as a collection of ``few'' temporal journeys and study scheduling and other combinatorial problems on this class of graphs.
    This natural parameter which, to the best of our knowledge, has not been studied yet, defines a structured family of temporal graphs. 
    Can this parameter allow for efficient algorithms in temporal graphs? If yes, then this will be a very positive exception in the literature.

\section{Acknowledgments}
This work was supported by 
the German Federal Ministry for Education and Research (BMBF) through the project ``KI Servicezentrum Berlin Brandenburg'' (01IS22092).
\bibliography{aaai25,temporal_graphs}


\end{document}